\documentclass[journal]{IEEEtran}
\ifCLASSINFOpdf
\else
\fi
\IEEEoverridecommandlockouts
\usepackage{setspace}
\usepackage{cite}
\usepackage{amsmath,amssymb,amsfonts,amsthm}

\newtheorem{lemma}{Lemma}
\newtheorem{proposition}{Proposition}
\usepackage{algorithmic}
\usepackage{algorithm}
\usepackage{graphicx}
\usepackage{comment}
\usepackage{textcomp}
\usepackage{color}
\usepackage{stfloats}
\usepackage{tabularx}
\usepackage{booktabs}
\usepackage[font=footnotesize]{caption}
\usepackage{subcaption}
\allowdisplaybreaks[2]

\usepackage{hyperref}

\begin{document}

\title{
Semantic Feature Multiple Access Empowered Integrated Learning and Communication Networks
}

\author{Jiaxiang Wang, Zhouxiang Zhao, Yahao Ding, Zhijin Qin, \IEEEmembership{Senior Member, IEEE}, Zhaohui Yang, 
\\ Mingzhe Chen, \IEEEmembership{Senior Member, IEEE}, 
and Mohammad Shikh-Bahaei,
\IEEEmembership{Senior Member, IEEE}
\thanks{Jiaxiang Wang, Yahao Ding and Mohammad Shikh-Bahaei are with the Centre for Telecommunications Research, Department of Engineering, King’s College London, WC2R 2LS London, U.K. (emails: jiaxiang.wang@kcl.ac.uk, yahao.ding@kcl.ac.uk, m.sbahaei@kcl.ac.uk).}
\thanks{Zhouxiang Zhao and Zhaohui Yang are with the College of Information Science and Electronic Engineering, Zhejiang University, Hangzhou 310027, China, and also with Zhejiang Provincial Key Lab of Information Processing, Communication and Networking (IPCAN), Hangzhou 310007, China. (emails: zhouxiangzhao@zju.edu.cn, yang\_zhaohui@zju.edu.cn). }
\thanks{Zhijin Qin is with the Department of Electronic Engineering, Tsinghua University, Beijing 100084, China. (e-mail: qinzhijin@tsinghua.edu.cn).}
\thanks{Mingzhe Chen is with the Department of Electrical and Computer Engineering and Institute for Data Science and Computing, University of Miami, Coral Gables, FL, 33146, USA. (e-mail: mingzhe.chen@miami.edu).
}
}

\maketitle
\begin{abstract}
Integrated learning and communication (ILAC) unifies learned transceivers with radio resource management, where semantic feature multiple access (SFMA) enables paired users to superpose their learned representations over shared time-frequency resources. Unlike conventional multiple access schemes, SFMA interference arises in the learned feature space and depends jointly on the user pair, the transmit power, and the compression ratio. This coupling ties binary pairing decisions to continuous resource variables, yielding a mixed-integer non-convex optimisation problem. To address this problem, we first propose similarity-conditioned SFMA (SC-SFMA), a Swin Transformer-based transceiver whose dual-conditioned similarity modulator (DC-SimM) gates cross-user feature fusion according to the inter-user semantic similarity. We then characterise the resulting pair-dependent interference by a bivariate logistic function parameterised by transmit power and compression ratio, thereby bridging the learned transceiver with network-level optimisation. On this basis, we formulate a sum-rate maximisation problem subject to per-user distortion, latency, energy, power, and bandwidth constraints. To solve this problem, we develop a three-block alternating optimisation algorithm that integrates dual-decomposition-assisted compression ratio allocation, trust-region successive convex approximation (SCA) for joint power–bandwidth optimisation, and dynamic feasible graph-based user pairing. Simulation results show that SC-SFMA achieves considerable peak signal-to-noise ratio (PSNR) and multi-scale structural similarity index measure (MS-SSIM) gains over deep joint source–channel coding (JSCC) and separation-based baselines. The proposed optimisation framework attains significant sum rate improvements than conventional multiple access baselines.

\end{abstract}

\section{Introduction}
 
Sixth-generation (6G) wireless networks are envisioned to natively embed artificial intelligence (AI) into the network fabric, evolving from data pipes into context-aware, task-driven infrastructures~\cite{saad2019vision,letaief2022edge,letaief2019roadmap}. This convergence has been formalised as the paradigm of {integrated learning and communication} (ILAC), which co-designs learned transceivers with radio resource management so that AI enhances transmission and, reciprocally, wireless networks support efficient model deployment~\cite{xu2023edge,xu2025ilac,zhang2025comai}. A prominent ILAC instantiation is {semantic communication}, which extracts and transmits only task-relevant features rather than raw bits, yielding substantial bandwidth savings for high-dimensional data~\cite{yang2022semantic,shi2023task}. As an ILAC-native scheme, semantic communication inherently couples the learned encoder--decoder pair with radio resource decisions such as power, bandwidth, and compression ratio~\cite{xu2025ilac,zhao2026agentic}. This coupling calls for a unified framework that simultaneously governs the learning and communication aspects of the network.
 
While deep joint source--channel coding (JSCC) has established the single-link foundation of semantic communication~\cite{bourtsoulatze2019deep,dai2022nonlinear,yang2024swinjscc}, scaling it to a multi-user network introduces a fundamentally new challenge. When multiple users share time--frequency resources, the superposition of their learned representations produces {semantic interference} whose severity depends not on power disparity, as in conventional non-orthogonal multiple access (NOMA), but on the semantic compatibility between paired users~\cite{wang2025semantic,wang2025generative}. This observation has motivated {semantic feature multiple access} (SFMA), which allows paired users to superpose their semantic features over shared resources~\cite{wang2025semantic}. SFMA introduces a three-way coupling absent in classical multiple access: the binary pairing decision, the continuous resource variables, and the pair-dependent semantic interference are mutually dependent. However, existing SFMA designs assume fixed user pairing without jointly optimising pairing decisions and resource allocation, leaving a unified ILAC framework for SFMA an open problem.
 
\subsection{Related Works}
 
A rich body of work has laid the groundwork for ILAC by investigating the interplay between learned transceivers and wireless resource management. Edge AI frameworks~\cite{letaief2022edge,letaief2019roadmap} advocate the co-design of communication and learning, edge learning with distributed signal processing~\cite{xu2023edge} establishes dual-functional performance metrics, and the recent ILAC formalisation~\cite{xu2025ilac} further unifies large AI models with adaptive communication. The convergence of communication and AI has also been systematised from the network architecture perspective~\cite{zhang2025comai}. These architectural visions have been concretised at the physical layer through deep JSCC~\cite{bourtsoulatze2019deep,dai2022nonlinear,yang2024swinjscc,erdemir2023generative} and its multi-user extensions, yet the question of how to jointly manage semantic transceivers and radio resources across users within an ILAC framework remains largely open.
 
Several recent works have begun to address multi-user semantic transmission. Hierarchical broadcasting~\cite{bo2024deep} delivers scalable quality to heterogeneous users, task-oriented multi-user communication~\cite{xie2022task} jointly optimises encoding for downstream inference, and NOMA-enhanced semantic transmission~\cite{li2023non} applies successive interference cancellation in the power domain. More recently, semantic-domain multiple access paradigms have emerged: DeepMA~\cite{zhang2024deepma} trains mutually orthogonal encoder--decoder pairs for channel multiplexing, semantic feature division multiple access (SFDMA)~\cite{ma2024semantic} encodes multi-user information into discrete approximately orthogonal subspaces, and SFMA~\cite{wang2025semantic,wang2025generative} superposes continuous semantic features of paired users for joint transmission. In the broader ILAC context, semantic-aware collaboration in wireless agent networks~\cite{zhao2026agentic} coordinates multi-agent sensing and communication. However, all these schemes either assume fixed user groupings or manage interference purely in the physical layer without jointly optimising the semantic pairing and radio resource allocation.
 
On the resource allocation side, semantic-aware optimisation has been studied from perspectives including cooperative multi-server management~\cite{zhang2023optimization}, rate splitting multiple access (RSMA)-based joint compression ratio and radio resource optimisation~\cite{zhao2025compression}, adaptable semantic compression~\cite{liu2023adaptable}, intelligent reflecting surface (IRS)-enhanced cross-layer allocation~\cite{wang2024irs}, and learning-based approaches~\cite{wang2022performance,zhang2023predictive}. Meanwhile, user pairing for conventional NOMA has been extensively studied based on channel gain disparities~\cite{ding2016impact}, dynamic clustering~\cite{ali2017dynamic}, and joint power allocation with decoding order~\cite{yang2023joint,zhu2018optimal}. A common limitation is that semantic resource allocation schemes assume orthogonal access and therefore do not encounter the pairing--resource coupling inherent in SFMA, while NOMA pairing criteria target physical-layer objectives and pair users by channel disparity rather than semantic similarity. In summary, a unified ILAC framework that jointly optimises user pairing and radio resource allocation by exploiting semantic-domain interference has not been developed.
 
\subsection{Contributions}
 
The main contributions of this paper are summarised as follows.
\begin{itemize}
 
\item We propose SC-SFMA, a Swin Transformer-based transceiver whose dual-conditioned similarity modulator (DC-SimM) adaptively gates cross-user feature fusion according to inter-user semantic similarity, enabling graceful interpolation between independent coding and full feature fusion within the ILAC framework.
 
\item We establish a pair-conditioned analytical framework that bridges the learned transceiver with radio resource management. The semantic interference is captured by a bivariate logistic function jointly parameterised by transmit power and compression ratio, and a sum-rate maximisation problem is formulated subject to per-user distortion, latency, energy, power, and bandwidth constraints.
 
\item We develop a three-block alternating optimisation algorithm comprising dual-decomposition-assisted compression ratio allocation, trust-region successive convex approximation (SCA) for joint power--bandwidth optimisation, and dynamic-feasible-graph-based user pairing, which produces a monotonically non-decreasing objective sequence with polynomial per-iteration complexity.
 
\item Simulation results demonstrate that SC-SFMA achieves significant peak signal-to-noise ratio (PSNR) and multi-scale structural similarity index measure (MS-SSIM) gains over deep JSCC and separation-based baselines, while the joint optimisation framework attains up to $39\%$ higher sum rate than NOMA-based transmission and $32\%$ higher than channel-gain-based pairing.
 
\end{itemize}
 
The remainder of this paper is organised as follows. Section~\ref{sec:system_model} presents the system model and problem formulation. Section~\ref{sec:solution} develops the proposed algorithm. Section~\ref{sec:results} provides simulation results. Section~\ref{sec:conclusion} concludes the paper.

\section{System Model}
\label{sec:system_model}

\begin{figure*}
    \centering
    \includegraphics[width=\linewidth]{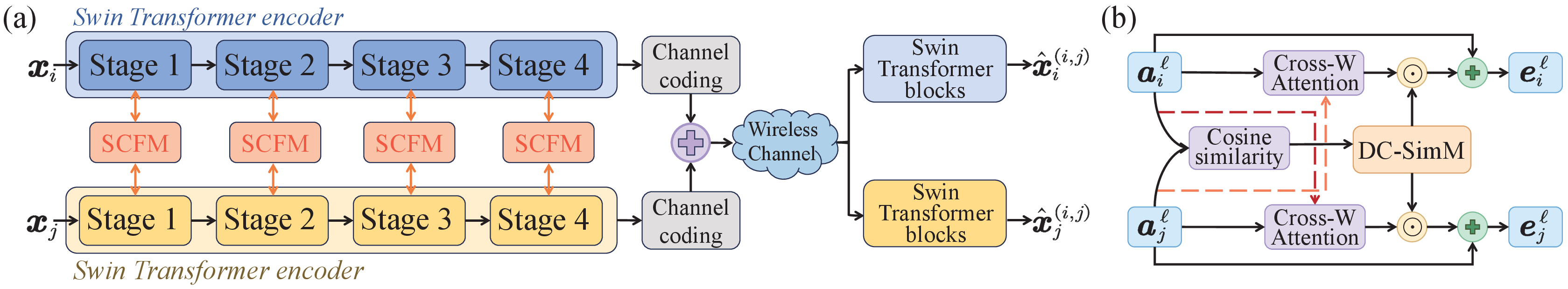}
    \caption{System architecture of the proposed SC-SFMA framework. (a) Overall transceiver: two users' images are encoded by four-stage Swin Transformer encoders with SCFM modules exchanging cross-user semantic features at each stage; the encoded features are superposed and transmitted over a shared wireless channel, and each user independently decodes the received signal via Swin Transformer blocks. (b) Detail of the SCFM module: the dual-conditioned similarity modulator computes per-channel modulation weights from the inter-user cosine similarity, which gate the cross-window attention output before residual addition.}
    \label{fig:system_model}
\end{figure*}

We consider a downlink semantic communication system where one base station (BS) serves an even number of users indexed by $\mathcal{N}=\{1,\dots,N\}$. The users are partitioned into $K=N/2$ two-user groups indexed by $\mathcal{K}=\{1,\dots,K\}$. Let 
\begin{equation}
\label{eq:edge set}
    \mathcal{P}=\{(i,j)\mid 1\le i<j\le N\}
\end{equation}
denote the set of all candidate user pairs. For each candidate pair $(i,j)\in\mathcal{P}$ and group $k\in\mathcal{K}$, we define the binary assignment variable
\begin{equation}
\label{eq:pair_indicator}
w_{ij}^k=
\begin{cases}
1, & \text{if pair $(i,j)$ is assigned to group $k$,}\\
0, & \text{otherwise.}
\end{cases}
\end{equation}
Unless otherwise stated, $w_{ij}^k$ is defined only for $i<j$ in order to avoid duplicated pair indices. Each active group $k$ is allocated one bandwidth $b_k$, transmit power $p_k$, and semantic compression ratio $\delta_k$. The two users within the same group share the same compression ratio $\delta_k$, while different groups may use different compression ratios. The BS employs the SFMA mechanism to fuse the semantic features of each selected user pair and then transmits the fused representation over a shared time-frequency resource block. We first describe the general SFMA method, then present the proposed SC-SFMA architecture. After that, we develop the communication, latency, energy, and distortion models used for resource allocation, and then culminate in the problem formulation.

\subsection{Semantic Feature Multiple Access}
\label{subsec:sfma}
For a given input image $\boldsymbol{x}_u \in \mathbb{R}^{H \times W \times 3}$, user $u\in\{i,j\}$, the semantic encoder $\boldsymbol{E}_{\boldsymbol{\psi}_u}(\cdot)$ extracts a feature representation
\begin{equation}
\label{eq:encoder}
{\boldsymbol{z}}_{u} = \boldsymbol{E}_{\boldsymbol{\psi}_u}(\boldsymbol{x}_u) \in \mathbb{R}^{H'\times W' \times C'},
\end{equation}
where $H'$, $W'$, $C'$ are the spatial and channel dimensions after encoding. A channel coding function $\mathcal{F}(\cdot)$ then maps $\boldsymbol{z}_u$ to a length-$d$ channel vector
\begin{equation}
\label{eq:channel_mapping}
\tilde{\boldsymbol{z}}_u = \mathcal{F}(\boldsymbol{z}_u, \Omega_u) \in \mathbb{R}^{d},
\end{equation}
where $\Omega_u$ carries side information shared between the BS and user $u$. For paired group $k$, the normalised signals $\boldsymbol{s}_u = \tilde{\boldsymbol{z}}_u / \sqrt{\mathbb{E}[\|\tilde{\boldsymbol{z}}_u\|_2^2]/d}$ of the two users are superposed as
\begin{equation}
\label{eq:tx_signal}
\boldsymbol{s}_k^{(i,j)} = \sqrt{{p_k}/{2}}\;(\boldsymbol{s}_i + \boldsymbol{s}_j),
\end{equation}
where $p_k$ is the transmit power for group $k$. The received signal at user $u$ in group $k$ is
\begin{equation}
\label{eq:rx_signal}
\boldsymbol{y}^{(i,j)}_{k,u} = h_u\,\boldsymbol{s}_k^{(i,j)} + \boldsymbol{n}_u,
\end{equation}
where $h_u$ denotes the channel gain from the BS to user $u$, and $\boldsymbol{n}_u \sim \mathcal{CN}(0, b_k N_0 \boldsymbol{I})$ is the additive white Gaussian noise, where $b_k$ is the bandwidth allocated to group $k$, $N_0$ is the noise power spectral density, and $\boldsymbol{I}$ is the identity matrix of matching dimension. The channel decoder $\mathcal{R}(\cdot)$ recovers the feature representation $\hat{\boldsymbol{z}}_u^{(i,j)} = \mathcal{R}(\boldsymbol{y}^{(i,j)}_{k,u}, \Omega_u) \in \mathbb{R}^{H' \times W' \times C'}$, and the semantic decoder $\boldsymbol{D}_{\boldsymbol{\phi}_u}(\cdot)$ reconstructs the image as
\begin{equation}
\label{eq:decoder}
\hat{\boldsymbol{x}}^{(i,j)}_u = \boldsymbol{D}_{\boldsymbol{\phi}_u}(\hat{\boldsymbol{z}}_u^{(i,j)}) \in \mathbb{R}^{H \times W \times 3}.
\end{equation}
 
\subsection{Proposed SC-SFMA Model}
\label{subsec:sc_sfma}
 
Building upon the SFMA framework, we propose the similarity-conditioned SFMA (SC-SFMA) model, whose architecture is illustrated in Fig.~\ref{fig:system_model}. The key components are described below.
 
\subsubsection{Hierarchical Encoder with SCFM}
The encoder follows a four-stage Swin Transformer~\cite{liu2021swin} hierarchy. Stage~$1$ partitions the input image into non-overlapping $2\times 2$ patches to form an initial feature $\boldsymbol{f}_u^0 \in \mathbb{R}^{\frac{H}{2}\times\frac{W}{2}\times C_1}$, where $C_1$ is the base embedding dimension. Therefore, the stage~$1$ feature has spatial size $H_1 \times W_1 = (H/2) \times (W/2)$. Each subsequent stage~$\ell\in\{2,3,4\}$ first applies patch merging, which halves the spatial resolution and doubles the channel dimension, transforming $\boldsymbol{f}_u^{\ell-1}$ into $\boldsymbol{g}_u^\ell \in \mathbb{R}^{H_\ell\times W_\ell\times C_\ell}$ with $H_\ell = H_{\ell-1}/2$, $W_\ell = W_{\ell-1}/2$, and $C_\ell = 2C_{\ell-1}$. At every stage, the features pass through a similarity-conditioned cross-attention feature modulation (SCFM) module, followed by $N_\ell$ Swin Transformer blocks, where $N_\ell$ denotes the number of blocks at stage $\ell$, producing the stage output $\boldsymbol{f}_u^\ell$. The final output $\boldsymbol{z}_u \triangleq \boldsymbol{f}_u^4 \in \mathbb{R}^{H'\times W'\times C'}$ with $H'=H/16$, $W'=W/16$, and $C'=8C_1$ is the semantic feature used in~\eqref{eq:encoder}.
 
The SCFM module is the core mechanism that makes SC-SFMA adaptive to inter-user semantic similarity. Let $\boldsymbol{a}_u^\ell$ denote the input to the SCFM module at stage $\ell$ ($\boldsymbol{f}_u^0$ for $\ell=1$ and $\boldsymbol{g}_u^\ell$ for $\ell\ge 2$). The module first computes a global cosine similarity $\alpha^\ell$ and a channel-wise mean absolute difference $\boldsymbol{\Delta}^\ell$ between the paired users' features. These two descriptors are fed into a dual-conditioned similarity modulator (DC-SimM), a three-layer fully-connected network that produces per-channel modulation weights $\boldsymbol{\gamma}^\ell \in (0,1)^{C_\ell}$ via a sigmoid output gate. Concurrently, the module applies windowed and shifted-window cross-attention~\cite{liu2021swin} between the paired users to obtain a cross-attention enhanced feature $\boldsymbol{e}_{u,\mathrm{cross}}^\ell$. The final SCFM output is
\begin{equation}
\label{eq:scfm_output}
\boldsymbol{e}_u^\ell = \boldsymbol{a}_u^\ell + \boldsymbol{\gamma}^\ell \odot (\boldsymbol{e}_{u,\mathrm{cross}}^\ell - \boldsymbol{a}_u^\ell),
\end{equation}
where $\odot$ denotes element-wise multiplication broadcast over the spatial dimension. When $\boldsymbol{\gamma}^\ell \to 0$ (low similarity), the output preserves the user's own features; when $\boldsymbol{\gamma}^\ell \to 1$ (high similarity), the output fully incorporates cross-user information.
 
\subsubsection{Rate-Adaptive Channel Coding and Decoding}
\label{subsubsec:channel_coding}
To support variable compression ratios, we adopt a Rate-ModNet mechanism~\cite{yang2024swinjscc}. Given a target compression ratio $\delta \in [\delta_{\min},1]$, a three-layer fully-connected network transforms $\delta$ into a per-channel rate modulation vector $\boldsymbol{r}_u \in (0,1)^{C'}$ via a sigmoid gate. The modulated feature $\boldsymbol{z}'_u = \boldsymbol{z}_u \odot \boldsymbol{r}_u$ is then processed by a code mask module that ranks channels by average activation magnitude and retains the top $C_s = \lfloor \delta \cdot C' \rfloor$ channels, producing a binary mask $\boldsymbol{M}_u \in \{0,1\}^{C'}$. The selected channels are compacted into $\tilde{\boldsymbol{Z}}_u \in \mathbb{R}^{H'\times W'\times C_s}$ and flattened to form the channel input $\tilde{\boldsymbol{z}}_u \in \mathbb{R}^d$ with $d = H'W'C_s$. The channel bandwidth ratio is accordingly $\mathrm{CBR} = d/(3HW)$.
 
At the receiver, the channel decoder $\mathcal{R}(\cdot)$ uses the known compression ratio $\delta$ and the mask side information $\Omega_u = \{\boldsymbol{M}_u\}$ to reconstruct the full-dimensional feature. In the online deployment stage, the compression ratio is instantiated by the group-specific value $\delta$ assigned to the scheduled pair.
 
\subsubsection{Semantic Decoding}
The decoder mirrors the encoder in reverse: four progressive stages apply Swin Transformer blocks followed by patch expansion to successively restore spatial resolution, generating the reconstructed image $\hat{\boldsymbol{x}}_u^{(i,j)}$ as in~\eqref{eq:decoder}. The complete framework is trained end-to-end by minimising the expected pairwise reconstruction loss:
\begin{equation}
\label{eq:training_loss}
\mathcal{L}_{\text{train}} = \mathbb{E}_{(\boldsymbol{x}_i,\boldsymbol{x}_j)}\!\left[\|\hat{\boldsymbol{x}}_i^{(i,j)} - \boldsymbol{x}_i\|_2^2 + \|\hat{\boldsymbol{x}}_j^{(i,j)} - \boldsymbol{x}_j\|_2^2\right].
\end{equation}

\subsection{Communication Model}
\label{subsec:comm_model}

Since the superposition is performed in the learned semantic-feature domain, the recovered feature of each user generally contains not only the desired semantic component but also the semantic component of the paired user. To capture this effect, we define the semantic interference factor $\rho_{ij}(p_k,\delta_k)\in[0,1]$ as a pair-level scaling factor \cite{wang2025generative}.
The transmit power $p_k$ determines the received SNR of the superposed signal, while the compression ratio $\delta_k$ determines the number of effective feature dimensions available for semantic separation at the decoder. 
Since this dependence is induced by the end-to-end learned transceiver, it is non-trivial to theoretically build an explicit form.
To overcome this obstacle, we employ the data regression method and observe from Fig.~\ref{fig:rho_3D} that: 1) $\rho_{ij}$~is monotonically non-increasing in both $p_k$ and $\delta_k$ and bounded within $[\rho^{\min}_{ij},\,\rho^{\max}_{ij}]\subseteq[0,1]$. 
2) The partial derivative with respect to each variable first increases and then decreases, exhibiting a sigmoid-like transition from~$\rho^{\max}_{ij}$ to~$\rho^{\min}_{ij}$.

\begin{figure}[t]
    \centering
    \includegraphics[width=0.9\linewidth]{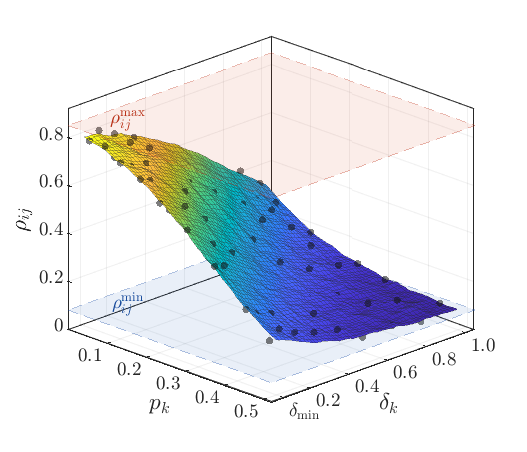}
    \caption{The semantic interference factor between users $i$ and $j$ in group $k$.}
    \label{fig:rho_3D}
\end{figure}

The above observations indicate that $\rho_{ij}$~follows a reverse~`S'-shaped surface over $(p_k,\delta_k)$.
We therefore adopt the following pair-dependent surrogate model
\begin{equation}
\label{eq:rho_power}
\rho_{ij}(p_k,\delta_k)\approx \rho_{ij}^{\min}
+\frac{\rho_{ij}^{\max}-\rho_{ij}^{\min}}
{1+\exp(a_{ij}p_k+b_{ij}\delta_k+d_{ij})},
\end{equation}
where $\rho_{ij}^{\min}$ and $\rho_{ij}^{\max}$ denote the lower and upper saturation levels of the semantic interference factor, respectively, $a_{ij}>0$ and $b_{ij}>0$ control the sensitivities to $p_k$ and $\delta_k$, and $d_{ij}$ is an offset parameter.
Increasing $\delta_k$ means that more semantic channels are preserved and the compression is lighter. Increasing $p_k$ preserves richer and more pair-discriminative semantic features, which also improves the separability of the two users. Therefore, the semantic interference factor $\rho_{ij}(p_k,\delta_k)$ is empirically non-increasing in both variables. The upper saturation $\rho_{ij}^{\max}$ corresponds to severely resource-limited regimes, while the lower saturation $\rho_{ij}^{\min}$ reflects the irreducible feature-space overlap imposed by the source distribution and model architecture.
Using $\rho_{ij}(p_k,\delta_k)$ and following \cite{wang2025generative}, the effective signal-to-interference-plus-noise ratio (SINR) of user $u$ on group $k$ is modelled as
\begin{equation}
\label{eq:sinr_model}
\Gamma_{ij,u}^{k}=
\frac{\frac{p_k}{2}|h_u|^2}
{\rho_{ij}(p_k,\delta_k)\frac{p_k}{2}|h_u|^2+b_kN_0},
\quad u\in\{i,j\}.
\end{equation}
The corresponding semantic-aware transmission rate is
\begin{equation}
R_{ij,u}^{k}(p_k,b_k,\delta_k)
=
b_k\log_2\!\left(1+\Gamma_{ij,u}^{k}\right),
\quad u\in\{i,j\},
\end{equation}
and the pair sum-rate on group $k$ is
\begin{equation}
\bar{R}_{ij}^{k}(p_k,b_k,\delta_k)
=
R_{ij,i}^{k}(p_k,b_k,\delta_k)
+
R_{ij,j}^{k}(p_k,b_k,\delta_k).
\end{equation}

\subsection{Latency and Energy Model}
\label{subsec:latency_energy}
Let $Q_u$ denote the source size of user $u\in\{i,j\}$. Since both users in group $k$ share the compression ratio $\delta_k$ which governs the Rate-ModNet mask for both users to implement the feature-domain superposition in \eqref{eq:tx_signal}, the transmission delay of user $u$ on group $k$ is
\begin{equation}
\label{eq:user_delay}
t_{ij,u}^k = \frac{Q_u\delta_k}{R_{ij,u}^k}.
\end{equation}
Let $\tau_{ij}^{\mathrm{BS}} = C_{ij}^{\mathrm{BS}}/f_{\mathrm{BS}}$ and $\tau_u^{\mathrm{dec}} = C_u^{\mathrm{dec}}/f_u$ denote the BS-side encoding latency and the user-side decoding latency, respectively, where $C_{ij}^{\mathrm{BS}}$ and $C_u^{\mathrm{dec}}$ are the corresponding computation loads, and $f_{\mathrm{BS}}$ and $f_u$ are the CPU frequencies. The end-to-end latency of pair $(i,j)$ on group $k$ is
\begin{equation}
\label{eq:e2e_latency}
T_{ij}^k
=
\tau_{ij}^{\mathrm{BS}}
+
\max\left\{t_{ij,i}^k+\tau_i^{\mathrm{dec}},\; t_{ij,j}^k+\tau_j^{\mathrm{dec}}\right\}.
\end{equation}
The total energy consumption of pair $(i,j)$ on group $k$ is
\begin{equation}
\label{eq:total_energy}
E_{ij}^k
=
p_k\max\left\{t_{ij,i}^k,\; t_{ij,j}^k\right\}
+
\zeta\ln\frac{1}{\delta_k},
\end{equation}
where the first term is the communication energy, and the second term is the computation energy induced by semantic compression.

\subsection{Distortion Model}
\label{subsec:distortion_model}
In the proposed SC-SFMA framework, the reconstruction quality of each user depends jointly on its paired user and the adopted compression ratio. Therefore, for any candidate pair $(i,j)\in\mathcal{P}$, any compression ratio $\delta\in[\delta_{\min},1]$, and either user $u\in\{i,j\}$, the reconstruction distortion is
\begin{equation}
\label{eq:pair_distortion_curve}
D_{u,ij}(\delta)
\triangleq
\mathbb{E}\!\left[\frac{\|\boldsymbol{x}_u-\hat{\boldsymbol{x}}_u^{(i,j)}(\delta)\|_2^2}{3HW}\right],
\quad u\in\{i,j\}.
\end{equation}
In the offline stage, the BS evaluates each candidate pair on a compression ratio grid $\mathcal{D}=\{\bar\delta_1,\ldots,\bar\delta_L\}\subseteq[\delta_{\min},1]$ and constructs a monotone nonincreasing piecewise-linear upper envelope $\widehat D_{u,ij}(\delta)$ for the empirical distortion samples. The envelope $\widehat D_{u,ij}(\delta)$ is used online because it preserves conservative distortion feasibility while keeping the optimisation tractable.
For each user $u\in\{i,j\}$ and pair $(i,j)$, the minimum compression ratio that satisfies the distortion requirement $D_u^{\max}$ is
\begin{equation}
\label{eq:distortion_inverse}
\delta_{u,ij}^{\mathrm D}
\triangleq
\inf\left\{\delta\in[\delta_{\min},1]\,\big|\,\widehat D_{u,ij}(\delta)\le D_u^{\max}\right\}.
\end{equation}
The pair-level distortion-feasible lower bound is then
\begin{equation}
\label{eq:distortion_lower_bound}
\underline\delta_{ij}^{\mathrm D}
\triangleq
\max\left\{\delta_{\min},\delta_{i,ij}^{\mathrm D},\delta_{j,ij}^{\mathrm D}\right\}.
\end{equation}

\subsection{Problem Formulation}
\label{subsec:problem_formulation}
Given the system model above, our objective is to maximise the total semantic-aware transmission rate of the SC-SFMA network while guaranteeing per-user reconstruction quality, end-to-end latency, total energy budget, and the power and bandwidth budgets. 
Mathematically, the sum-rate maximisation problem can be formulated as
\begin{subequations}\label{eq:P1}
\begin{align}
\max_{\mathbf{w},\mathbf{p},\mathbf{b},\boldsymbol{\delta}} \;
& \sum_{k=1}^{K}\sum_{(i,j)\in\mathcal{P}} w_{ij}^k \bar R_{ij}^k
\label{eq:P1_obj} \tag{\theequation},
\\
\text{s.t.} \hspace{0.8em}
& \sum_{k=1}^{K}\sum_{(i,j)\in\mathcal{P}:u\in\{i,j\}} w_{ij}^k \widehat D_{u,ij}(\delta_k) \le D_u^{\max}, \forall u\in\mathcal{N},
\label{eq:P1_distortion}
\\
& \sum_{(i,j)\in\mathcal{P}} w_{ij}^k T_{ij}^k \le T^{\max}, \quad \forall k\in\mathcal{K},
\label{eq:P1_latency}
\\
& \sum_{k=1}^{K}\sum_{(i,j)\in\mathcal{P}} w_{ij}^k E_{ij}^k \le E^{\max},
\label{eq:P1_energy}
\\
& \sum_{k=1}^{K} p_k \le P^{\max},
\label{eq:P1_power}
\\
& \sum_{k=1}^{K} b_k \le B^{\max},
\label{eq:P1_bandwidth}
\\
& \delta_{\min}\le \delta_k \le 1,\quad \forall k\in\mathcal{K},
\label{eq:P1_delta}
\\
& \sum_{(i,j)\in\mathcal{P}} w_{ij}^k = 1,\quad \forall k\in\mathcal{K},
\label{eq:P1_group_pair}
\\
& \sum_{k=1}^{K}\sum_{(i,j)\in\mathcal{P}:u\in\{i,j\}} w_{ij}^k = 1,
\quad \forall u\in\mathcal{N},
\label{eq:P1_user_once}
\\
& w_{ij}^k\in\{0,1\},\quad \forall (i,j)\in\mathcal{P},\ \forall k\in\mathcal{K},
\label{eq:P1_binary}
\\
& p_k\ge 0,\; b_k\ge 0,\quad \forall k\in\mathcal{K},
\label{eq:P1_nonnegative}
\end{align}
\end{subequations}
where $\mathbf{p}=[p_1,\ldots,p_K]^\top$, $\mathbf{b}=[b_1,\ldots,b_K]^\top$, $\boldsymbol{\delta}=[\delta_1,\ldots,\delta_K]^\top$, and $\mathbf{w}=\{w_{ij}^k\}_{(i,j)\in\mathcal{P},k\in\mathcal{K}}$. 

Constraint \eqref{eq:P1_distortion} enforces the distortion requirement of each user. 
Constraint \eqref{eq:P1_latency} imposes the per-group end-to-end latency budget. Constraint \eqref{eq:P1_energy} limits the total energy consumption, while \eqref{eq:P1_power} and \eqref{eq:P1_bandwidth} impose the power and bandwidth budgets. Constraint \eqref{eq:P1_delta} bounds the admissible compression ratio of each group, and \eqref{eq:P1_group_pair} together with \eqref{eq:P1_user_once} enforces one-to-one pairing across all groups.

\section{Algorithm Design}
\label{sec:solution}
Problem~\eqref{eq:P1} is a mixed-integer non-convex problem. It is difficult to solve directly because the binary pairing variables are coupled with the continuous resource variables through the semantic-interference surface in \eqref{eq:rho_power}, while the compression ratio further affects the rate, latency, energy, and distortion terms simultaneously. To obtain a tractable solution, we adopt a three-block alternating-optimisation framework that updates the compression ratio allocation, the power-bandwidth allocation, and the user pairing in turn. Before presenting these online updates, we first describe the offline profiling and feasibility-pruning steps, which convert the learned SC-SFMA behaviour into tractable pair-level functions and remove candidate pairs that can never appear in any feasible solution.

The offline stage profiles the distortion envelopes $\{\widehat D_{u,ij}(\cdot)\}$ and the semantic-interference samples $\{\widehat \rho_{ij}(\bar p_n,\bar\delta_\ell)\}$, and then fits the logistic parameters in \eqref{eq:rho_power}. Using the distortion-feasible lower bounds in \eqref{eq:distortion_lower_bound} and the residual time budgets in \eqref{eq:time_budget}, we prune permanently infeasible pairs through
\begin{equation}
\label{eq:feasible_edge}
\mathcal{E}_{\mathrm f}
=
\left\{
(i,j)\in\mathcal{P}\,\big|\,
\underline\delta_{ij}^{\mathrm D}\le 1,\;
\bar T_{ij,i}>0,\;
\bar T_{ij,j}>0
\right\}.
\end{equation}
Only edges in $\mathcal{E}_{\mathrm f}$ are retained in the online optimisation stage. The following proposition confirms that restricting the candidate pool to $\mathcal{E}_{\mathrm f}$ incurs no loss of optimality.

\begin{proposition}
\label{prop:edge_equivalence}
Removing all candidate pairs outside $\mathcal{E}_{\mathrm f}$ does not change the feasible set of problem~\eqref{eq:P1}.
\end{proposition}

\begin{proof}
For any $(i,j)\notin\mathcal{E}_{\mathrm f}$, either $\underline\delta_{ij}^{\mathrm D}>1$, which means no admissible compression ratio can satisfy \eqref{eq:P1_distortion}, or $\bar T_{ij,u}\le 0$ for some $u\in\{i,j\}$, which means \eqref{eq:P1_latency} cannot be met. Hence no feasible solution of \eqref{eq:P1} can use an edge outside $\mathcal{E}_{\mathrm f}$.
\end{proof}

With the feasible edge set $\mathcal{E}_{\mathrm{f}}$ established, we now develop the three-block alternating-optimisation algorithm that solves problem~\eqref{eq:P1} over the reduced candidate set. 


\subsection{Compression Ratio Sub-problem}
\label{subsec:delta_solution}
We first optimise the compression ratio block under fixed feasible pairing $\mathbf{w}$ and radio resources $(\mathbf{p},\mathbf{b})$. Under this condition, the only remaining coupling across groups comes from the total energy budget. This makes the compression ratio block particularly suitable for dual decomposition: once the global energy constraint is priced by a single multiplier, the optimisation separates into $K$ scalar sub-problems, each depending only on $\delta_k$.
Let $(i_k,j_k)$ be the unique pair assigned to group $k$, that is, $w_{i_kj_k}^k=1$. 
The compression ratio sub-problem is formulated as
\begin{subequations}
\label{eq:Pdelta}
\begin{align}
\max_{\boldsymbol{\delta}}\quad
& \sum_{k=1}^{K}\bar R_{i_kj_k}^k(p_k,b_k,\delta_k)
\label{eq:Pdelta_obj} \tag{\theequation},
\\
\text{s.t.}\quad
& \underline\delta_{i_kj_k}^{\mathrm D}\le \delta_k\le 1,\quad \forall k\in\mathcal{K},
\label{eq:Pdelta_range}
\\
& T_{i_kj_k}^k(p_k,b_k,\delta_k)\le T^{\max},\quad \forall k\in\mathcal{K},
\label{eq:Pdelta_latency}
\\
& \sum_{k=1}^{K}E_{i_kj_k}^k(p_k,b_k,\delta_k)\le E^{\max}.
\label{eq:Pdelta_energy}
\end{align}
\end{subequations}

The only remaining cross-group coupling in \eqref{eq:Pdelta} lies in the total energy constraint \eqref{eq:Pdelta_energy}. We therefore dualise this constraint so that the compression ratio update decomposes into independent group-wise searches while the multiplier $\lambda\ge 0$ coordinates the global energy budget across groups. Therefore, the corresponding Lagrangian is
\begin{align}
\label{eq:delta_lagrangian}
\mathcal{L}_{\delta}(\boldsymbol{\delta},\lambda)  &= \lambda \left(E^{\max}- \sum_{k=1}^{K}E_{i_kj_k}^k(p_k,b_k,\delta_k)\right)\notag \\
& +\sum_{k=1}^{K}\bar R_{i_kj_k}^k(p_k,b_k,\delta_k).
\end{align}
For fixed $\lambda$, the variables $\{\delta_k\}$ decouple across groups, and group $k$ solves the one-dimensional problem
\begin{equation}
\label{eq:Pdelta_k}
\max_{\substack{\underline\delta_{i_kj_k}^{\mathrm D}\le \delta_k\le 1\\ T_{i_kj_k}^k(p_k,b_k,\delta_k)\le T^{\max}}}
\bar R_{i_kj_k}^k(p_k,b_k,\delta_k)-\lambda E_{i_kj_k}^k(p_k,b_k,\delta_k).
\end{equation}
Since \eqref{eq:Pdelta_k} is one-dimensional over a closed and bounded feasible interval, any interior maximiser must satisfy the first-order stationarity condition, which can be stated in semi-closed form. Let $\phi_k = \exp(a_{i_kj_k}p_k+b_{i_kj_k}\delta_k+d_{i_kj_k})$, and calculate the logistic derivative of $\rho_{i_kj_k}$ with respect to $\delta_k$ as $\rho^\prime_{i_kj_k} = -\frac{b_{i_kj_k}(\rho^{\max}_{i_kj_k}-\rho^{\min}_{i_kj_k})\phi_k}{(1+\phi_k)^2}<0$. The derivative of the per-user rate with respect to $\delta_k$ evaluates to $\frac{\partial R^k_{i_kj_k,u}}{\partial \delta_k} = \frac{-b_k{p^2_k h^4_u} \rho^\prime_{i_kj_k}}{\ln2(2b_kN_0+\rho_{i_kj_k}{p_k h^2_u})(2b_kN_0+\rho_{i_kj_k}{p_k h^2_u})}, u\in\{i_k,j_k\}$. Let $u^\star=\arg \max_{u\in\{i_k,j_k\}}t^k_{i_kj_k,u}$ denote the latency-dominant user. Differentiating the objective of \eqref{eq:Pdelta_k} and equating to zero yields the stationarity condition:
\begin{equation} \label{eq:newton}
\sum_{u} 
\frac{\partial R_{i_k j_k, u}^k}{\partial \delta_k}=\lambda \left[ \frac{p_kQ_{u^*}(R_{i_k j_k, u^*}^k - \delta_k \frac{\partial R_{i_k j_k, u^*}^k}{\partial \delta_k})}{\left( R_{i_k j_k, u^*}^k\right)^2}-\frac{\zeta}{\delta_k}\right],
\end{equation}
where $u\in\{i_k,j_k\}$. Because \eqref{eq:newton} is not amenable to a further closed-form solution in $\delta_k$, Newton's method is applied to calculate $\delta_k$ for each user group.



After obtaining the per-group maximisers for a given $\lambda$, the multiplier is updated by the subgradient step
\begin{equation}
\label{eq:delta_dual_update}
\lambda^{(s+1)}
=
\left[
\lambda^{(s)}
+
\beta_{\delta}^{(s)}
\left(
\sum_{k=1}^{K}E_{i_kj_k}^k(p_k,b_k,\delta_k^{(s)})-E^{\max}
\right)
\right]^+,
\end{equation}
where $\beta_{\delta}^{(s)}>0$ is the step size. Algorithm~\ref{alg:delta_update} summarises the optimisation procedure of sub-problem \eqref{eq:Pdelta}.

\begin{algorithm}[t]
\caption{Semantic Compression Ratio Optimization}
\label{alg:delta_update}
\begin{algorithmic}[1]
\STATE \textbf{Input:} fixed pairing $\mathbf{w}$, fixed radio resources $(\mathbf{p},\mathbf{b})$, initial multiplier $\lambda^{(0)}\ge 0$, step-size sequence $\{\beta_{\delta}^{(s)}\}$, search tolerance $\varepsilon_{\delta}$, and coarse grid size $M_{\delta}$.
\FOR{$s=0,1,\ldots$}
    \FOR{each $k\in\mathcal{K}$}
        \STATE Determine the selected pair $(i_k,j_k)$ from $\mathbf{w}$.
        \STATE Identify the feasible interval of $\delta_k$ from \eqref{eq:Pdelta_k}.
        \IF{the feasible interval is empty}
            \STATE \textbf{return} infeasible.
        \ENDIF
        \STATE Generate $M_{\delta}$ feasible samples for $\delta_k$ and choose the best one as the initial point.
        \STATE Refine the point by Newton's method on \eqref{eq:Pdelta_k} to tolerance $\varepsilon_{\delta}$, and obtain $\delta_k^{(s+1)}$.
    \ENDFOR
    \STATE Update $\lambda^{(s+1)}$ by \eqref{eq:delta_dual_update}.
    \IF{$\left|\sum_{k=1}^{K}E_{i_kj_k}^k(p_k,b_k,\delta_k^{(s+1)})-E^{\max}\right|<\varepsilon_{\delta}$}
        \STATE \textbf{break}
    \ENDIF
\ENDFOR
\STATE \textbf{Output:} $\boldsymbol{\delta}$.
\end{algorithmic}
\end{algorithm}

\subsection{Power-Bandwidth Sub-problem}
\label{subsec:pb_solution}
Fix the pairing $\mathbf{w}$ and the compression ratio vector $\boldsymbol{\delta}$, and let $(i_k,j_k)$ be the pair assigned to group $k$.
We next optimise the power-bandwidth block. 
In this case, the latency-induced rate floors become fixed thresholds, but the objective and the rate constraints remain non-convex because the semantic-interference term $\rho_{i_k j_k}(p_k,\delta_k)$ still depends non-linearly on $p_k$. Therefore, a direct convex reformulation is no longer available.
From the delay model in \eqref{eq:user_delay}, we define the residual transmission-time budget of user $u\in\{i_k,j_k\}$ under pair $(i_k,j_k)$ as
\begin{equation}
\label{eq:time_budget}
\bar T_{i_kj_k,u} \triangleq T^{\max}-\tau_{i_kj_k}^{\mathrm{BS}}-\tau_u^{\mathrm{dec}}.
\end{equation}
Then, the latency constraint for user $u$ becomes 
\begin{equation}
R_{i_kj_k,u}^k(p_k,b_k,\delta_k)\ge \frac{Q_u\delta_k}{\bar T_{i_kj_k,u}}.
\end{equation}
Moreover, the communication-energy term in \eqref{eq:total_energy} is upper-bounded by $p_k\max\{\bar T_{i_kj_k,i_k},\bar T_{i_kj_k,j_k}\}$. This yields the following tractable power-bandwidth sub-problem:
\begin{subequations}
\label{eq:Ppb}
\begin{align}
\max_{\mathbf{p},\mathbf{b}}
& \sum_{k=1}^{K}\Big(R_{i_kj_k,i_k}^k(p_k,b_k,\delta_k)+R_{i_kj_k,j_k}^k(p_k,b_k,\delta_k)\Big)
\label{eq:Ppb_obj} \tag{\theequation},
\\
\text{s.t.}\quad
& R_{i_kj_k,i_k}^k(p_k,b_k,\delta_k)\ge \frac{Q_{i_k}\delta_k}{\bar T_{i_kj_k,i_k}},\quad \forall k\in\mathcal{K},
\label{eq:Ppb_rate_i}
\\
& R_{i_kj_k,j_k}^k(p_k,b_k,\delta_k)\ge \frac{Q_{j_k}\delta_k}{\bar T_{i_kj_k,j_k}},\quad \forall k\in\mathcal{K},
\label{eq:Ppb_rate_j}
\\
& \sum_{k=1}^{K}\left(p_k\max\{\bar T_{i_kj_k,i_k},\bar T_{i_kj_k,j_k}\}+\zeta\ln\frac{1}{\delta_k}\right)\le E^{\max},
\label{eq:Ppb_energy}
\\
& \eqref{eq:P1_power}, \eqref{eq:P1_bandwidth}, \eqref{eq:P1_nonnegative}. \notag
\label{eq:Ppb_nonnegative}
\end{align}
\end{subequations}

To construct a local convex surrogate, we aim to isolate the non-linear dependence on $p_k$ brought by $\rho_{i_k j_k}(p_k,\delta_k)$. For fixed $\delta_k$, the pair-dependent term $\rho_{i_k j_k}(p_k,\delta_k)$ becomes a univariate non-linear function of $p_k$. Since it appears in the SINR through $q_k(p_k)=p_k\rho_{i_k j_k}(p_k,\delta_k)$, we introduce this product explicitly and linearise it around the current iterate. 
Using \eqref{eq:sinr_model}, the exact rate of user $u\in\{i_k,j_k\}$ is rewritten as
\begin{equation}
\label{eq:rate_q_form}
 R_{i_kj_k,u}^k(p_k,b_k,\delta_k) = b_k\log_2 \frac{N_0+\frac{p_k h_u^2}{2b_k}(1+\rho_{i_kj_k}(p_k,\delta_k))}{N_0+\frac{p_k h_u^2}{2b_k}\rho_{i_kj_k}(p_k,\delta_k)} .
\end{equation}
The resulting optimisation problem \eqref{eq:Ppb_obj} is non-convex. We therefore adopt a SCA strategy, in which the non-linear term $q_k(p_k)$ and the non-convex part of the rate expression are replaced by locally tight first-order models around the current iterate. Since this approximation is only reliable locally, we embed the update into a trust-region mechanism.
At the $m$-th inner SCA iteration, we linearly approximate the product $p_k\rho_{i_kj_k}(p_k,\delta_k)$ around $p_k^{(m)}$ as
\begin{equation}
\label{eq:qk_taylor}
\begin{aligned}
\widehat q_k^{(m)}(p_k)
&=
p_k^{(m)}\rho_{i_kj_k}(p_k^{(m)},\delta_k)
\\
&+
\left(
\rho_{i_kj_k}(p_k^{(m)},\delta_k) +p_k^{(m)}\iota
\right)
(p_k-p_k^{(m)}),
\end{aligned}
\end{equation}
where $\iota=\frac{\partial \rho_{i_kj_k}(p_k,\delta_k)}{\partial p_k}|_{p_k=p_k^{(m)}}$.
Substituting \eqref{eq:qk_taylor} into the two logarithmic terms in \eqref{eq:rate_q_form} gives
\begin{equation}
\label{eq:rplus_hat}
\widehat r_{i_kj_k,u}^{+,(m)}(p_k,b_k)
=
\frac{b_k}{\ln 2}\ln\left(N_0+\frac{|h_u|^2}{2b_k}\Big[p_k+\widehat q_k^{(m)}(p_k)\Big]\right),
\end{equation}
and
\begin{equation}
\label{eq:rminus_hat}
\widehat r_{i_kj_k,u}^{-,(m)}(p_k,b_k)
=
\frac{b_k}{\ln 2}\ln\left(N_0+\frac{|h_u|^2}{2b_k}\widehat q_k^{(m)}(p_k)\right).
\end{equation}
Both functions are concave over any trust region where $\widehat q_k^{(m)}(p_k)\ge 0$. Linearising the second concave term at $(p_k^{(m)},b_k^{(m)})$ yields the surrogate lower bound
\begin{align}
\label{eq:rate_lower_bound}
\underline R_{i_kj_k,u}^{k,(m)}(p_k,&b_k)
=
\widehat r_{i_kj_k,u}^{+,(m)}(p_k,b_k)
-
\widehat r_{i_kj_k,u}^{-,(m)}\!\left(p_k^{(m)},b_k^{(m)}\right)
\nonumber\\
&-
\nabla \widehat r_{i_kj_k,u}^{-,(m)}\!\left(p_k^{(m)},b_k^{(m)}\right)^\top
\begin{bmatrix}
p_k-p_k^{(m)}\\
b_k-b_k^{(m)}
\end{bmatrix}.
\end{align}
Let the objective value of problem~\eqref{eq:P1} be $\Xi(\mathbf{w},\mathbf{p},\mathbf{b},\boldsymbol{\delta})
\triangleq \sum_{k=1}^{K}\sum_{(i,j)\in\mathcal{P}} w_{ij}^k\bar R_{ij}^k$.
Correspondingly, we define the surrogate objective value for each iteration $m$ as
$\widehat\Xi^{(m)}(\mathbf{p},\mathbf{b})\triangleq \sum_{k=1}^{K}\Big(\underline R_{i_kj_k,i_k}^{k,(m)}+\underline R_{i_kj_k,j_k}^{k,(m)}\Big)$.
Hence, at the $m$-th iteration, we solve the following convex surrogate problem:
\begin{subequations}
\label{eq:Ppb_cvx}
\begin{align}
\max_{\mathbf{p},\mathbf{b}}\quad
& \widehat\Xi^{(m)}(\mathbf{p},\mathbf{b})
\label{eq:Ppb_cvx_obj} \tag{\theequation},
\\
\text{s.t.}\quad
& \underline R_{i_kj_k,i_k}^{k,(m)} \ge \frac{Q_{i_k}\delta_k}{\bar T_{i_kj_k,i_k}},\quad \forall k\in\mathcal{K},
\label{eq:Ppb_cvx_rate_i}
\\
& \underline R_{i_kj_k,j_k}^{k,(m)} \ge \frac{Q_{j_k}\delta_k}{\bar T_{i_kj_k,j_k}},\quad \forall k\in\mathcal{K},
\label{eq:Ppb_cvx_rate_j}
\\
& |p_k-p_k^{(m)}|\le \Delta_k^{(m)},\quad \forall k\in\mathcal{K},
\label{eq:Ppb_cvx_trust}
\\
& \widehat q_k^{(m)}(p_k)\ge 0,\quad \forall k\in\mathcal{K},
\label{eq:Ppb_cvx_qhat_nonnegative}
\\
& \eqref{eq:Ppb_energy}, \eqref{eq:P1_bandwidth}, \eqref{eq:P1_nonnegative}. \notag
\end{align}
\end{subequations}

\begin{lemma}
\label{lem:pb_convex}
For any fixed inner iterate $(\mathbf{p}^{(m)},\mathbf{b}^{(m)})$ and trust-region radii $\{\Delta_k^{(m)}\}$, the problem~\eqref{eq:Ppb_cvx} is convex.
\end{lemma}

\begin{proof}
Because $\widehat q_k^{(m)}(p_k)$ in \eqref{eq:qk_taylor} is affine in $p_k$, both \eqref{eq:rplus_hat} and \eqref{eq:rminus_hat} are perspectives of the concave logarithm composed with positive affine arguments, and are therefore concave in $(p_k,b_k)$ over the affine domain enforced by \eqref{eq:Ppb_cvx_qhat_nonnegative}. The first-order Taylor expansion of the concave function $\widehat r_{i_kj_k,u}^{-,(m)}$ is an affine upper bound. Hence the surrogate rate $\underline R_{i_kj_k,u}^{k,(m)}$ in \eqref{eq:rate_lower_bound} is concave. The objective in \eqref{eq:Ppb_cvx_obj} is therefore concave. Constraints \eqref{eq:Ppb_cvx_rate_i} and \eqref{eq:Ppb_cvx_rate_j} are superlevel sets of concave functions and are convex, while the remaining constraints are affine. Therefore, problem~\eqref{eq:Ppb_cvx} is convex.
\end{proof}

\begin{algorithm}[t]
\caption{Power and Bandwidth Optimisation}
\label{alg:pb_update}
\begin{algorithmic}[1]
\STATE \textbf{Input:} fixed pairing $\mathbf{w}$, fixed $\boldsymbol{\delta}$, initial $(\mathbf{p}^{(0)},\mathbf{b}^{(0)})$, radii $\{\Delta_k^{(0)}\}_{k\in\mathcal{K}}$, and $(\eta_1,\eta_2,\kappa_{\mathrm{sh}},\kappa_{\mathrm{ex}},\varepsilon_{\mathrm{pb}})$.
\FOR{$m=0,1,\ldots$}
    \STATE Form \eqref{eq:Ppb_cvx} using \eqref{eq:qk_taylor} and \eqref{eq:rate_lower_bound}, and solve it to obtain $(\widehat{\mathbf{p}},\widehat{\mathbf{b}})$.
    \STATE $\Delta_{\mathrm{pred}}^{(m)}\leftarrow \widehat\Xi^{(m)}(\widehat{\mathbf{p}},\widehat{\mathbf{b}})-\widehat\Xi^{(m)}(\mathbf{p}^{(m)},\mathbf{b}^{(m)})$.
    \IF{$\Delta_{\mathrm{pred}}^{(m)}\le 0$}
        \STATE \textbf{break}
    \ENDIF
    \STATE Compute the trust-region ratio $\eta^{(m)}$ from \eqref{eq:trust_ratio}.
    \IF{$(\widehat{\mathbf{p}},\widehat{\mathbf{b}})$ is feasible for problem~\eqref{eq:P1} and $\eta^{(m)}\ge\eta_1$}
        \STATE Accept $(\mathbf{p}^{(m+1)},\mathbf{b}^{(m+1)})\leftarrow(\widehat{\mathbf{p}},\widehat{\mathbf{b}})$.
        \IF{$\eta^{(m)}\ge\eta_2$}
            \STATE $\Delta_k^{(m+1)}\leftarrow \kappa_{\mathrm{ex}}\Delta_k^{(m)}$, $\forall k\in\mathcal{K}$.
        \ELSE
            \STATE $\Delta_k^{(m+1)}\leftarrow \Delta_k^{(m)}$, $\forall k\in\mathcal{K}$.
        \ENDIF
        \IF{$\dfrac{\Xi(\mathbf{w},\mathbf{p}^{(m+1)},\mathbf{b}^{(m+1)},\boldsymbol{\delta})-\Xi(\mathbf{w},\mathbf{p}^{(m)},\mathbf{b}^{(m)},\boldsymbol{\delta})}{\max\{1,|\Xi(\mathbf{w},\mathbf{p}^{(m)},\mathbf{b}^{(m)},\boldsymbol{\delta})|\}}<\varepsilon_{\mathrm{pb}}$}
            \STATE \textbf{break}
        \ENDIF
    \ELSE
        \STATE Reject the candidate and set $(\mathbf{p}^{(m+1)},\mathbf{b}^{(m+1)})\leftarrow(\mathbf{p}^{(m)},\mathbf{b}^{(m)})$.
        \STATE $\Delta_k^{(m+1)}\leftarrow \kappa_{\mathrm{sh}}\Delta_k^{(m)}$, $\forall k\in\mathcal{K}$.
    \ENDIF
\ENDFOR
\STATE \textbf{Output:} last accepted $(\mathbf{p},\mathbf{b})$.
\end{algorithmic}
\end{algorithm}

After solving the surrogate problem \eqref{eq:Ppb_cvx} and obtaining a candidate $(\widehat{\mathbf{p}},\widehat{\mathbf{b}})$, we compare the predicted improvement of the surrogate objective with the actual improvement of the original objective. It is quantified by the trust-region ratio
\begin{equation}
\label{eq:trust_ratio}
\eta^{(m)}
\triangleq
\frac{\Xi(\mathbf{w},\widehat{\mathbf{p}},\widehat{\mathbf{b}},\boldsymbol{\delta})-\Xi(\mathbf{w},\mathbf{p}^{(m)},\mathbf{b}^{(m)},\boldsymbol{\delta})}{\widehat\Xi^{(m)}(\widehat{\mathbf{p}},\widehat{\mathbf{b}})-\widehat\Xi^{(m)}(\mathbf{p}^{(m)},\mathbf{b}^{(m)})},
\end{equation}
whenever the denominator is positive. It is then used to decide whether the candidate step should be accepted and whether the trust-region radii should be expanded or contracted. 
The power-bandwidth update thus alternates between solving the convex surrogate problem and validating the resulting candidate against the original objective. Algorithm~\ref{alg:pb_update} summarises the complete optimisation procedure of sub-problem \eqref{eq:Ppb_obj}. 
The initial point $(\mathbf{p}^{(0)}, \mathbf{b}^{(0)})$ for Algorithm \ref{alg:pb_update} is set to the equal-allocation point $p^{(0)}_k=P^{\max}/K$, $b^{(0)}_k=B^{\max}/K$ at the first outer iteration, and to the last accepted continuous solution $(\mathbf{p}^{n},\mathbf{b}^{n})$ from the previous outer iteration thereafter.

\subsection{User Pairing Sub-problem}
\label{subsec:pairing_solution}
We finally update the user pairing under fixed $(\mathbf{p},\mathbf{b},\boldsymbol{\delta})$. Thus, all pair-level rates, latency terms, energy costs, and distortion costs become constants. The remaining difficulty is therefore no longer nonlinear resource coupling, but the combinatorial one-to-one assignment structure together with the globally coupled user-exclusivity, energy, and distortion constraints.
For each group $k$, we first screen the static edge set $\mathcal{E}_{\mathrm f}$ using the current continuous variables and retain only those pairs that satisfy the current distortion lower bound and latency-induced rate floors on that group. The resulting dynamic feasible edge set is defined as
\begin{equation}
\label{eq:dynamic_feasible_edge}
\mathcal{E}_{\mathrm f}^k(\mathbf{p},\mathbf{b},\boldsymbol{\delta})
=
\left\{(i,j)\in\mathcal{E}_{\mathrm f}\,\middle|\,
\begin{aligned}
&\delta_k\ge \underline\delta_{ij}^{\mathrm D},\\
&R_{ij,i}^k(p_k,b_k,\delta_k)\ge \frac{Q_i\delta_k}{\bar T_{ij,i}},\\
&R_{ij,j}^k(p_k,b_k,\delta_k)\ge \frac{Q_j\delta_k}{\bar T_{ij,j}}
\end{aligned}
\right\}.
\end{equation}
Then, the user pairing sub-problem is reformulated as
\begin{subequations}
\label{eq:Pw}
\begin{align}
\max_{\mathbf{w}}\quad
& \sum_{k=1}^{K}\sum_{(i,j)\in\mathcal{E}_{\mathrm f}} w_{ij}^k\bar R_{ij}^k(p_k,b_k,\delta_k)
\label{eq:Pw_obj} \tag{\theequation},
\\
\text{s.t.}\quad
& \sum_{(i,j)\in\mathcal{E}_{\mathrm f}} w_{ij}^k=1,\quad \forall k\in\mathcal{K},
\label{eq:Pw_group}
\\
& \sum_{k=1}^{K}\sum_{(i,j)\in\mathcal{E}_{\mathrm f}:u\in\{i,j\}} w_{ij}^k\le 1,\quad \forall u\in\mathcal{N},
\label{eq:Pw_user}
\\
& \sum_{k=1}^{K}\sum_{(i,j)\in\mathcal{E}_{\mathrm f}} w_{ij}^kE_{ij}^k(p_k,b_k,\delta_k)\le E^{\max},
\label{eq:Pw_energy}
\\
& \sum_{k=1}^{K}\sum_{(i,j)\in\mathcal{E}_{\mathrm f}:u\in\{i,j\}} w_{ij}^k\widehat D_{u,ij}(\delta_k)\le D_u^{\max}, \; \forall u\in\mathcal{N},
\label{eq:Pw_distortion}
\\
& w_{ij}^k=0,\quad \forall (i,j)\in\mathcal{E}_{\mathrm f}\setminus \mathcal{E}_{\mathrm f}^k(\mathbf{p},\mathbf{b},\boldsymbol{\delta}),\ \forall k\in\mathcal{K},
\label{eq:Pw_dynamic_zero}
\\
& w_{ij}^k\in\{0,1\},\quad \forall (i,j)\in\mathcal{E}_{\mathrm f}, \forall k\in\mathcal{K}.
\label{eq:Pw_binary}
\end{align}
\end{subequations}

Problem \eqref{eq:Pw} is still difficult because, although the objective is now linear in the binary variables, the constraints in \eqref{eq:Pw_user}–\eqref{eq:Pw_distortion} couple the group-wise assignment decisions globally. To decouple these interactions, we relax the binary constraints and dualise the user-exclusivity, total-energy, and per-user distortion constraints. This converts the original pairing problem into a dual pricing problem in which each candidate edge is evaluated according to its net utility under the current prices. Thus, we relax \eqref{eq:Pw_binary} to $w_{ij}^k\in[0,1]$ and dualise \eqref{eq:Pw_user}, \eqref{eq:Pw_energy}, and \eqref{eq:Pw_distortion}. Let $\nu_u\ge 0$, $\vartheta\ge 0$, and $\lambda_u^{\mathrm D}\ge 0$ be the corresponding multipliers, and $\boldsymbol{\nu}=[\nu_1,\cdots,\nu_N]$, $\boldsymbol{\lambda}^{\mathrm D}=[{\lambda}^{\mathrm D}_1,\cdots,{\lambda}^{\mathrm D}_N]$. We define the dual price as $\pi_{ij}^k(p_k,b_k,\delta_k)=\bar R_{ij}^k(p_k,b_k,\delta_k)
-\vartheta E_{ij}^k(p_k,b_k,\delta_k)
-\lambda_i^{\mathrm D}\widehat D_{i,ij}(\delta_k)
-\lambda_j^{\mathrm D}\widehat D_{j,ij}(\delta_k)$. 
Under the dual prices, each candidate pair-group assignment is assigned a priced utility consisting of its pair sum-rate minus the penalties induced by the global energy budget and the users’ distortion budgets.
The pair-group utility is defined as
\begin{equation}
\label{eq:pair_utility}
\Phi_{ij}^k
=
\begin{cases}
\pi_{ij}^k(p_k,b_k,\delta_k),
&\text{if }(i,j)\in\mathcal{E}_{\mathrm f}^k(\mathbf{p},\mathbf{b},\boldsymbol{\delta}),\\
-\infty, &\text{otherwise.}
\end{cases}
\end{equation}
Accordingly, the net dual utility of assigning pair $(i,j)$ to group $k$ is
\begin{equation}
\label{eq:reduced_cost}
\varpi_{ij}^k=\Phi_{ij}^k-\nu_i-\nu_j.
\end{equation}
Substituting these priced utilities into the relaxed problem yields the Lagrangian
\begin{align}
\label{eq:pair_lagrangian}
\mathcal{L}_{\mathrm w}
=  \sum_{k=1}^{K}\sum_{(i,j)\in\mathcal{E}_{\mathrm f}}
 w_{ij}^k\varpi_{ij}^k
+\sum_{u\in\mathcal{N}} (\nu_u + \lambda_u^{\mathrm D}D_u^{\max})
+\vartheta E^{\max},
\end{align}
in which the remaining inter-group interaction is carried only through the dual prices.
The following lemma shows that, once the coupled constraints are priced, the relaxed pairing update for each group reduces to selecting the dynamically feasible edge with the largest reduced cost.
\begin{lemma}
\label{lem:pair_relaxation}
For fixed multipliers $(\boldsymbol{\nu},\vartheta,\boldsymbol{\lambda}^{\mathrm D})$, the relaxed pairing sub-problem of each group $k$ admits the one-hot solution
\begin{equation}
\label{eq:relaxed_pairing_solution}
w_{ij}^{k,\star}=
\begin{cases}
1, & \text{if }(i,j)=\arg\max\limits_{(m,n)\in\mathcal{E}_{\mathrm f}^k(\mathbf{p},\mathbf{b},\boldsymbol{\delta})}\varpi_{mn}^k,\\
0, & \text{otherwise.}
\end{cases}
\end{equation}
\end{lemma}

\begin{proof}
For fixed $k$, the relaxed constraints reduce to $\sum_{(i,j)\in\mathcal{E}_{\mathrm f}}w_{ij}^k=1$ and $0\le w_{ij}^k\le 1$, so the feasible set is a simplex. The group-wise objective extracted from \eqref{eq:pair_lagrangian} is linear in $\{w_{ij}^k\}_{(i,j)\in\mathcal{E}_{\mathrm f}}$, and only dynamically feasible edges have finite reduced costs. Hence the optimum is attained at an extreme point of the simplex, namely the one-hot vector associated with the largest reduced cost among the dynamically feasible edges.
\end{proof}

The one-hot rule in \eqref{eq:relaxed_pairing_solution} solves the group-wise dual-relaxed selection under the current prices. The prices are then adjusted by projected sub-gradient steps so as to penalise user conflicts, excessive total energy, and excessive per-user distortion in subsequent iterations.
They are updated by
\begin{subequations}
\begin{align}
  \nu_u^{(t+1)}
    &= \left[
         \nu_u^{(t)}
         + \mu^{(t)}\!\left(
             \sum_{k=1}^{K}\sum_{(i,j)\in\mathcal{E}_{\mathrm{f}}:u\in\{i,j\}}
             w_{ij}^{k,(t)}-1
           \right)
       \right]^{+}\!, \label{eq:pair_dual_u} \\[3pt]
  \vartheta^{(t+1)}
    &= \left[\vartheta^{(t)}\right. \notag \\
    + \left.\beta_E^{(t)}\right. & \left.\left(
             \sum_{k=1}^{K}\sum_{(i,j)\in\mathcal{E}_{\mathrm{f}}}
             w_{ij}^{k,(t)}E_{ij}^k(p_k,b_k,\delta_k)-E^{\max}
           \right)\right]^{+}\!, \label{eq:pair_dual_energy} \\[3pt]
  \lambda_u^{\mathrm{D},(t+1)}
    &= \Bigg[\lambda_u^{\mathrm{D},(t)}\Bigg. \notag \\
    + \Bigg.\beta_D^{(t)} & \Bigg(
           \sum_{k=1}^{K}\sum_{(i,j)\in\mathcal{E}_{\mathrm{f}}:u\in\{i,j\}}
             w_{ij}^{k,(t)}\widehat{D}_{u,ij}(\delta_k)-D_u^{\max}
           \Bigg)\Bigg]^{+}\!. \label{eq:pair_dual_distortion}
\end{align}
\end{subequations}
Here, $\mu^{(t)}$, $\beta_E^{(t)}$, and $\beta_D^{(t)}$ are positive diminishing step sizes. 
More specifically, $\nu_u$ prices repeated use of user $u$, $\vartheta$ prices the total energy budget, and $\lambda^{\mathrm D}_u$ prices the distortion budget of user $u$.
Although the dual iterations drive the reduced costs toward respecting the coupled constraints on average, the resulting one-hot selections may still violate primal user exclusivity because different groups make their selections independently. We therefore perform a two-stage procedure after the dual loop, consisting of conflict resolution followed by residual completion.

\emph{Stage 1:} 
The first stage removes user conflicts greedily while preserving the most valuable incident assignment for each conflicted user.
Each user keeps only the incident assignment with the largest reduced cost, and all conflicting assignments are removed. Let $\mathcal{U}_{\mathrm{free}}$ and $\mathcal{K}_{\mathrm{free}}$ denote the unmatched users and unresolved groups after this step. Define the residual edge sets
$\mathcal{E}_{\mathrm f}(\mathcal{U}_{\mathrm{free}})
\triangleq
\left\{(i,j)\in\mathcal{E}_{\mathrm f}\,\big|\, i,j\in\mathcal{U}_{\mathrm{free}}\right\}$
and $\mathcal{E}_{\mathrm f}^k(\mathcal{U}_{\mathrm{free}})
\triangleq
\left\{(i,j)\in\mathcal{E}_{\mathrm f}^k(\mathbf{p},\mathbf{b},\boldsymbol{\delta})\,\big|\, i,j\in\mathcal{U}_{\mathrm{free}}\right\}$.

\emph{Stage 2:} 
After conflict resolution, the remaining unmatched users and unresolved groups are completed through a residual matching-and-assignment step, which restores the one-to-one pairing structure on the residual graph.
Over the residual graph induced by $\mathcal{U}_{\mathrm{free}}$, define $\bar\varpi_{ij}=\max_{k\in\mathcal{K}_{\mathrm{free}}}\varpi_{ij}^k$. We first solve the residual maximum-weight perfect matching
\begin{subequations}
\label{eq:Pmatch}
\begin{align}
\max_{\{\chi_{ij}\}}\quad
& \sum_{(i,j)\in\mathcal{E}_{\mathrm f}(\mathcal{U}_{\mathrm{free}})}\bar\varpi_{ij}\chi_{ij}
\label{eq:Pmatch_obj} \tag{\theequation},
\\
\text{s.t.}\quad
& \sum_{(i,j)\in\mathcal{E}_{\mathrm f}(\mathcal{U}_{\mathrm{free}}):u\in\{i,j\}}\chi_{ij}=1,\quad \forall u\in\mathcal{U}_{\mathrm{free}},
\label{eq:Pmatch_user}
\\
& \chi_{ij}\in\{0,1\},\quad \forall (i,j)\in\mathcal{E}_{\mathrm f}(\mathcal{U}_{\mathrm{free}}),
\label{eq:Pmatch_binary}
\end{align}
\end{subequations}
and then assign the matched pairs to the unresolved groups through
\begin{subequations}
\label{eq:Passign}
\begin{align}
\max_{\{x_{mk}\}}\quad
& \sum_{m=1}^{|\mathcal{K}_{\mathrm{free}}|}\sum_{k\in\mathcal{K}_{\mathrm{free}}}x_{mk}\varpi_{i_mj_m}^k
\label{eq:Passign_obj} \tag{\theequation},
\\
\text{s.t.}\quad
& \sum_{k\in\mathcal{K}_{\mathrm{free}}}x_{mk}=1,\quad \forall m,
\label{eq:Passign_pair}
\\
& \sum_{m=1}^{|\mathcal{K}_{\mathrm{free}}|}x_{mk}=1,\quad \forall k\in\mathcal{K}_{\mathrm{free}},
\label{eq:Passign_group}
\\
& x_{mk}\in\{0,1\},\quad \forall m,\ \forall k\in\mathcal{K}_{\mathrm{free}}.
\label{eq:Passign_binary}
\end{align}
\end{subequations}
If either the residual matching or the residual assignment is infeasible, we fall back to the compact residual repair problem
\begin{subequations}
\label{eq:Prepair}
\begin{align}
\max_{\{\omega_{ij}^k\}}
& \sum_{k\in\mathcal{K}_{\mathrm{free}}}\sum_{(i,j)\in\mathcal{E}_{\mathrm f}^k(\mathcal{U}_{\mathrm{free}})} \omega_{ij}^k\varpi_{ij}^k
\label{eq:Prepair_obj} \tag{\theequation},
\\
\text{s.t.}\;
& \sum_{(i,j)\in\mathcal{E}_{\mathrm f}^k(\mathcal{U}_{\mathrm{free}})} \omega_{ij}^k=1,\quad \forall k\in\mathcal{K}_{\mathrm{free}},
\label{eq:Prepair_group}
\\
& \sum_{k\in\mathcal{K}_{\mathrm{free}}}\sum_{(i,j)\in\mathcal{E}_{\mathrm f}^k(\mathcal{U}_{\mathrm{free}}):u\in\{i,j\}} \omega_{ij}^k=1, \forall u\in\mathcal{U}_{\mathrm{free}},
\label{eq:Prepair_user}
\\
& \omega_{ij}^k\in\{0,1\}, \forall k\in\mathcal{K}_{\mathrm{free}},\ \forall (i,j)\in\mathcal{E}_{\mathrm f}^k(\mathcal{U}_{\mathrm{free}}).
\label{eq:Prepair_binary}
\end{align}
\end{subequations}
Finally, a pair-swap local refinement is executed: if swapping the assigned pairs of two groups preserves feasibility and increases the total pairing utility, the swap is accepted; the process terminates when no improving swap exists.
The validity of the above repair step is stated in Proposition~\ref{prop:projection_feasible}.

\begin{algorithm}[t]
\caption{User Pairing Optimisation}
\label{alg:pairing_update}
\begin{algorithmic}[1]
\STATE \textbf{Input:} current $(\mathbf{p},\mathbf{b},\boldsymbol{\delta})$, feasible edge set $\mathcal{E}_{\mathrm f}$, and dual variables $(\boldsymbol{\nu},\vartheta,\boldsymbol{\lambda}^{\mathrm D})$.
\REPEAT
    \STATE Construct $\{\mathcal{E}_{\mathrm f}^k(\mathbf{p},\mathbf{b},\boldsymbol{\delta})\}_{k\in\mathcal{K}}$ via \eqref{eq:dynamic_feasible_edge}.
    \STATE Compute $\Phi_{ij}^k$ and $\varpi_{ij}^k$ from \eqref{eq:pair_utility}--\eqref{eq:reduced_cost}.
    \STATE For each group $k$, choose the relaxed one-hot assignment by \eqref{eq:relaxed_pairing_solution}.
    \STATE Update $(\boldsymbol{\nu},\vartheta,\boldsymbol{\lambda}^{\mathrm D})$ via \eqref{eq:pair_dual_u}--\eqref{eq:pair_dual_distortion}.
\UNTIL{the dual stopping criterion is satisfied}
\STATE Resolve user conflicts by keeping, for each user, the incident assignment with the largest reduced cost.
\STATE On the residual subset, solve \eqref{eq:Pmatch} and \eqref{eq:Passign} when feasible; otherwise solve \eqref{eq:Prepair}.
\STATE Apply feasible pair swaps until no improving swap exists.
\STATE \textbf{Output:} feasible pairing $\widehat{\mathbf{w}}$ and updated dual variables $(\boldsymbol{\nu},\vartheta,\boldsymbol{\lambda}^{\mathrm D})$.
\end{algorithmic}
\end{algorithm}

\begin{proposition}
\label{prop:projection_feasible}
The repair procedure composed of conflict resolution, residual matching-assignment when feasible, and the fallback repair problem~\eqref{eq:Prepair} otherwise, returns a binary one-to-one pairing satisfying \eqref{eq:Pw_group}, \eqref{eq:Pw_user}, and \eqref{eq:Pw_binary} whenever the residual sub-problem over $(\mathcal{U}_{\mathrm{free}},\mathcal{K}_{\mathrm{free}})$ is feasible.
\end{proposition}

\begin{proof}
After conflict resolution, each retained assignment is binary and satisfies user exclusivity by construction. If Problems~\eqref{eq:Pmatch} and \eqref{eq:Passign} are feasible, they return disjoint residual pairs and assign them one-to-one to the unresolved groups, thereby satisfying \eqref{eq:Pw_group} and \eqref{eq:Pw_user}. If either of them is infeasible, problem~\eqref{eq:Prepair} directly imposes one pair per unresolved group and one use per unresolved user. Therefore, whenever the residual problem is feasible, combining the accepted assignments from Stage~1 with the Stage~2 or fallback solution yields a binary one-to-one pairing satisfying \eqref{eq:Pw_group}, \eqref{eq:Pw_user}, and \eqref{eq:Pw_binary}.
\end{proof}

The complete optimisation procedure of user pairing sub-problem \eqref{eq:Pw} is summarised in Algorithm~\ref{alg:pairing_update}.

\subsection{Algorithm Analysis}
\label{subsec:algorithm_analysis}

We now integrate Algorithms~\ref{alg:delta_update}--\ref{alg:pairing_update}
into the unified alternating-optimisation framework of Algorithm~\ref{alg:overall}.
Let $\Xi^{(n)} \triangleq \Xi\!\left(\mathbf{w}^{(n)},
\mathbf{p}^{(n)},\mathbf{b}^{(n)},\boldsymbol{\delta}^{(n)}\right)$
denote the accepted objective value at the $n$-th outer iteration.

\subsubsection{Convergence Analysis}

Denote the intermediate objective values after the compression ratio,
power--bandwidth, and pairing updates at iteration $n$ by
$\Xi^{(n)}_{c}$, $\Xi^{(n)}_{pb}$, and $\Xi^{(n)}_{w}$, respectively.
The three block updates satisfy
\begin{equation}
  \Xi^{(n-1)} \leq \Xi^{(n)}_{c} \leq \Xi^{(n)}_{pb}
              \leq \Xi^{(n)}_{w} = \Xi^{(n)},
  \label{eq:monotone_chain}
\end{equation}
establishing that $\{\Xi^{(n)}\}$ is monotonically non-decreasing.
Each inequality follows from the optimality of the respective block
update: Algorithm~\ref{alg:delta_update} maximises the Lagrangian
over $\boldsymbol{\delta}$; Algorithm~\ref{alg:pb_update} accepts
a candidate only when the trust-region ratio criterion~\eqref{eq:trust_ratio}
confirms a non-negative actual improvement; and the pairing candidate
is retained only after re-optimising the continuous variables and
verifying $\Xi^{(n)}_{w} \geq \Xi^{(n)}_{pb}$.
Since the sum rate is non-negative and the feasible set is compact
by constraints~\eqref{eq:P1_power}--\eqref{eq:P1_delta},
the sequence $\{\Xi^{(n)}\}$ is upper bounded and therefore converges.

\subsubsection{Complexity Analysis}

The per-iteration costs of the three sub-problems are:
$\mathcal{O}(I_{\lambda} K I_{\delta})$ for the compression ratio
update (Algorithm~\ref{alg:delta_update}, $I_\lambda$ dual iterations
each solving $K$ Newton searches);
$\mathcal{O}(I_{\mathrm{pb}} K^{3})$ for the power--bandwidth update
(Algorithm~\ref{alg:pb_update}, interior-point method over $2K$ variables);
and $\mathcal{O}(I_{\mathrm{w}} K|\mathcal{E}_{\mathrm{f}}|
+ |\mathcal{U}_{\mathrm{free}}|^{3} + K^{2})$ for the pairing update
(Algorithm~\ref{alg:pairing_update}, comprising dual iterations,
Blossom-type matching, and pair-swap search).
Since Algorithms~\ref{alg:delta_update} and~\ref{alg:pb_update}
run twice per outer iteration, the total complexity over
$I_{\mathrm{AO}}$ iterations is
\begin{equation}
  \mathcal{O}\!\left(I_{\mathrm{AO}}\!\left[
    I_{\lambda} K I_{\delta} + I_{\mathrm{pb}} K^{3}
    + I_{\mathrm{w}} K\lvert\mathcal{E}_{\mathrm{f}}\rvert
    + \lvert\mathcal{U}_{\mathrm{free}}\rvert^{3} + K^{2}
  \right]\right),
  \label{eq:total_complexity}
\end{equation}
where the $\mathcal{O}(I_{\mathrm{pb}}K^3)$ interior-point solve
dominates in practice, as offline pruning~\eqref{eq:dynamic_feasible_edge}
renders $|\mathcal{E}_{\mathrm{f}}|$ and $|\mathcal{U}_{\mathrm{free}}|$
small.

\begin{algorithm}[t]
\caption{Overall Optimisation}
\label{alg:overall}
\begin{algorithmic}[1]
\STATE \textbf{Offline phase:} profile $\{\widehat D_{u,ij}(\cdot)\}$ and $\{\widehat\rho_{ij}(\bar p_n,\bar\delta_\ell)\}$ on $\mathcal{G}_p\times\mathcal{D}$, fit the parameters in \eqref{eq:rho_power}, and construct $\mathcal{E}_{\mathrm f}$ via \eqref{eq:feasible_edge}.
\STATE Initialise a feasible tuple $(\mathbf{w}^{(0)},\mathbf{p}^{(0)},\mathbf{b}^{(0)},\boldsymbol{\delta}^{(0)})$ and pairing dual variables $(\boldsymbol{\nu}^{(0)},\vartheta^{(0)},\boldsymbol{\lambda}^{\mathrm D,(0)})$.
\FOR{$n=0,1,\ldots$}
    \STATE Run Algorithm~\ref{alg:delta_update} under $(\mathbf{w}^{(n)},\mathbf{p}^{(n)},\mathbf{b}^{(n)})$ to obtain $\boldsymbol{\delta}^{\mathrm c}$.
    \STATE Run Algorithm~\ref{alg:pb_update} under $(\mathbf{w}^{(n)},\boldsymbol{\delta}^{\mathrm c})$, initialised at $(\mathbf{p}^{(n)},\mathbf{b}^{(n)})$, and obtain $(\mathbf{p}^{\mathrm c},\mathbf{b}^{\mathrm c})$.
    \STATE Set $\Xi_{\mathrm c}^{(n+1)}=\Xi(\mathbf{w}^{(n)},\mathbf{p}^{\mathrm c},\mathbf{b}^{\mathrm c},\boldsymbol{\delta}^{\mathrm c})$.
    \STATE Run Algorithm~\ref{alg:pairing_update} with $(\mathbf{p}^{\mathrm c},\mathbf{b}^{\mathrm c},\boldsymbol{\delta}^{\mathrm c},\boldsymbol{\nu}^{(n)},\vartheta^{(n)},\boldsymbol{\lambda}^{\mathrm D,(n)})$ to obtain $(\widehat{\mathbf{w}}^{(n+1)},\boldsymbol{\nu}^{(n+1)},\vartheta^{(n+1)},\boldsymbol{\lambda}^{\mathrm D,(n+1)})$.
    \STATE Re-run Algorithm~\ref{alg:delta_update} and Algorithm~\ref{alg:pb_update} under $\widehat{\mathbf{w}}^{(n+1)}$, initialised at $(\mathbf{p}^{\mathrm c},\mathbf{b}^{\mathrm c})$; if feasible, obtain $(\widehat{\boldsymbol{\delta}},\widehat{\mathbf{p}},\widehat{\mathbf{b}})$ and set $\widehat\Xi^{(n+1)}=\Xi(\widehat{\mathbf{w}}^{(n+1)},\widehat{\mathbf{p}},\widehat{\mathbf{b}},\widehat{\boldsymbol{\delta}})$.
    \IF{the refined candidate is feasible for problem~\eqref{eq:P1} and $\widehat\Xi^{(n+1)}\ge \Xi_{\mathrm c}^{(n+1)}$}
        \STATE Accept $(\mathbf{w}^{(n+1)},\mathbf{p}^{(n+1)},\mathbf{b}^{(n+1)},\boldsymbol{\delta}^{(n+1)},\Xi^{(n+1)})\leftarrow(\widehat{\mathbf{w}}^{(n+1)},\widehat{\mathbf{p}},\widehat{\mathbf{b}},\widehat{\boldsymbol{\delta}},\widehat\Xi^{(n+1)})$.
    \ELSE
        \STATE Keep $(\mathbf{w}^{(n+1)},\mathbf{p}^{(n+1)},\mathbf{b}^{(n+1)},\boldsymbol{\delta}^{(n+1)},\Xi^{(n+1)})\leftarrow(\mathbf{w}^{(n)},\mathbf{p}^{\mathrm c},\mathbf{b}^{\mathrm c},\boldsymbol{\delta}^{\mathrm c},\Xi_{\mathrm c}^{(n+1)})$.
    \ENDIF
    \IF{$|\Xi^{(n+1)}-\Xi^{(n)}|<\varepsilon$}
        \STATE \textbf{break}
    \ENDIF
\ENDFOR
\STATE \textbf{Output:} final feasible tuple $(\mathbf{w},\mathbf{p},\mathbf{b},\boldsymbol{\delta})$.
\end{algorithmic}
\end{algorithm}

\section{Simulation Results}
\label{sec:results}

In this section, we present simulation results to validate the effectiveness of the proposed SC-SFMA model and the joint user pairing and resource allocation algorithm. 
 
\subsection{Simulation Setup}
\label{subsec:sim_setup}

In the simulations, $N$ users are uniformly distributed within a circular cell of radius 250~m and communicate with the BS. For the path loss model between each user and the BS, we use the model $128.1 + 37.6 \lg d$, where $d$ is in km, and set the standard deviation of shadow fading to $4$~dB. Unless otherwise stated, the default simulation parameters are listed in Table~\ref{tab:params}. All optimisation results are averaged over 100 independent Monte Carlo realisations of user locations and channel coefficients. The SC-SFMA model is trained end-to-end using the Adam optimiser with a learning rate of $10^{-4}$ and a batch size of $32$ for $100$ epochs. The training is conducted over a range of SNR values from $1$ to $13~\text{dB}$ to ensure robust performance across different channel conditions. 
For comparison, the proposed SC-SFMA is benchmarked against the following six schemes.

\begin{table}[t]
\renewcommand\arraystretch{1.2}
\footnotesize
\centering
\caption{Main system parameters}
\begin{tabular}{|c||c|}
    \toprule\hline
    \textbf{Parameter}  & \textbf{Value} \\
    \hline
    Cell radius & $250$~m \\ \hline
    Number of users $N$ & $10$ \\ \hline
    Total transmit power $P^{\max}$ & $30$~dBm \\ \hline
    Total bandwidth $B^{\max}$ & $10$~MHz \\ \hline
    Maximum latency $T^{\max}$ & $100$~ms \\ \hline
    Noise PSD $N_0$ & $-174$~dBm/Hz \\ \hline
    Base embedding dimension $C_1$ & $128$ \\ \hline
    $[N_1,N_2,N_3,N_4]$ & $[2,2,6,2]$ \\ \hline
    Minimum compression ratio $\delta_{\min}$ & $0.0625$ \\ \hline
    Distortion threshold $D_u^{\max}$ & $0.005$ \\ \hline
    BS CPU frequency $f_{\text{BS}}$ & $10$~GHz \\ \hline
    User CPU frequency $f_u$ & $1$~GHz \\ \hline
    Computation energy coefficient $\zeta$ & $5 \times 10^{-3}$~J \\
    \hline\bottomrule
\end{tabular}
\label{tab:params}
\end{table}

\begin{enumerate}
\item \textbf{SwinJSCC+NOMA}:
We employ a SwinJSCC encoder without SCFM module for each user and implement NOMA within each user group. The overall optimisation procedure follows the same procedure as Algorithm \ref{alg:overall}.
 
\item \textbf{SwinJSCC+OMA}:
Each user independently employs SwinJSCC and implements frequency division multiple access (FDMA) over an orthogonal sub-band with $b_u\!=\!B^{\max}\!/N$ and $p_u\!=\!P^{\max}\!/N$, without user pairing.
 
\item \textbf{DeepJSCC-NOMA}:
We employ the DeepJSCC~\cite{bourtsoulatze2019deep} model for each user with NOMA power-domain superposition.
 
\item \textbf{BPG+LDPC+QAM}:
This baseline implements BPG compression \cite{bellard2014bpg}, low-density parity-check (LDPC) channel coding, and quadrature amplitude modulation (QAM), in both OMA and NOMA schemes, and chooses the best-performing configuration of coding rate and modulation based on the adaptive modulation and coding (AMC) standard \cite{3gpp_38214_v1500}.
 
\item \textbf{Channel-based pairing}:
This baseline implements the same SC-SFMA transceiver as the proposed scheme, but users are paired by the classical NOMA criterion~\cite{ding2016impact}. The resource optimisation follows the same procedure as Algorithms~\ref{alg:delta_update} and \ref{alg:pb_update}.
 
\item \textbf{Equal allocation}:
This baseline implements the SC-SFMA transceiver with the proposed pairing algorithm \ref{alg:pairing_update}, but power and bandwidth are equally divided as $p_k\!=\!P^{\max}\!/K$ and $b_k\!=\!B^{\max}\!/K$ for all groups. The compression ratio is optimised via Algorithm~\ref{alg:delta_update}.
\end{enumerate}
 
Schemes~$1)-4)$ compare different transmission models under the same resource budget, isolating the model-level gain of the SCFM module and the advantage of JSCC over separated coding. Schemes~$5)-6)$ share the SC-SFMA transceiver and profiled $\rho_{ij}$ tables with the proposed method, isolating the impact of different pairing and resource allocation strategies.

\subsection{SC-SFMA Model Performance}
\label{subsec:model_perf}

\begin{figure}[t]

    \centering
    
    \begin{subfigure}[t]{0.49\linewidth}
        \centering
        \includegraphics[width=\linewidth]{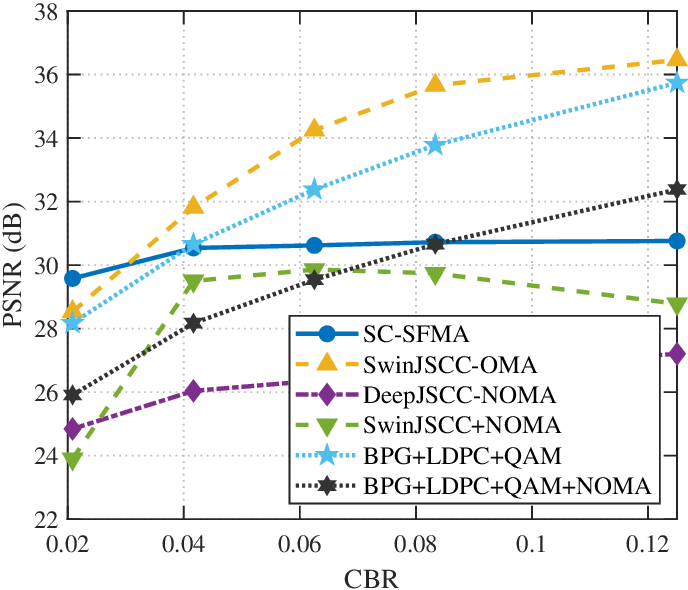}
        \caption{}
        \label{fig:sub1}
    \end{subfigure}
    \hfill
    \begin{subfigure}[t]{0.49\linewidth}
        \centering
        \includegraphics[width=\linewidth]{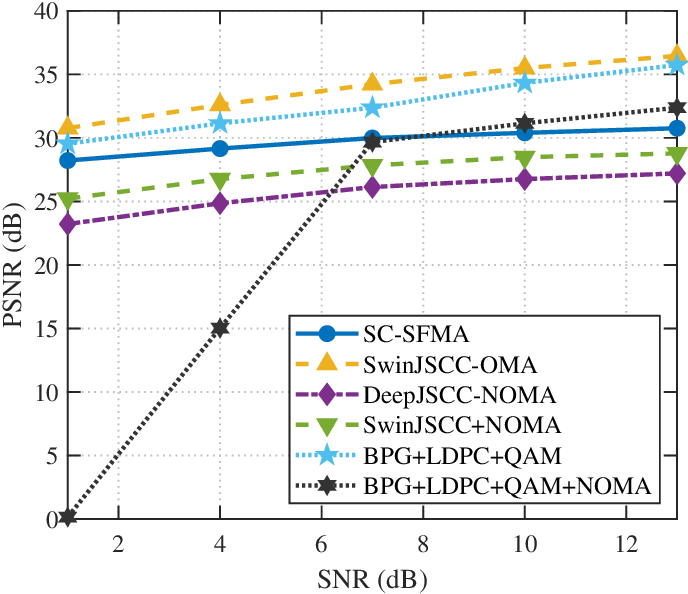}
        \caption{}
        \label{fig:sub2}
    \end{subfigure}
    
    \vspace{0.3cm}
    
    \begin{subfigure}[t]{0.49\linewidth}
        \centering
        \includegraphics[width=\linewidth]{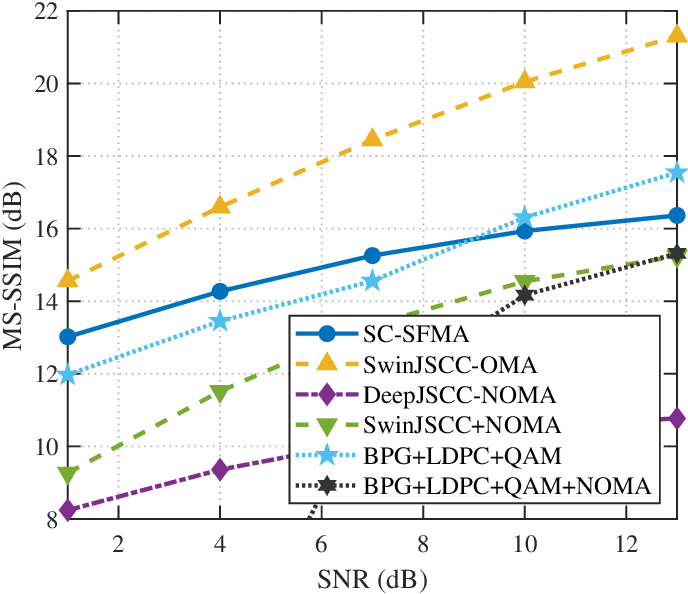}
        \caption{}
        \label{fig:sub3}
    \end{subfigure}
    \hfill
    \begin{subfigure}[t]{0.49\linewidth}
        \centering
        \includegraphics[width=\linewidth]{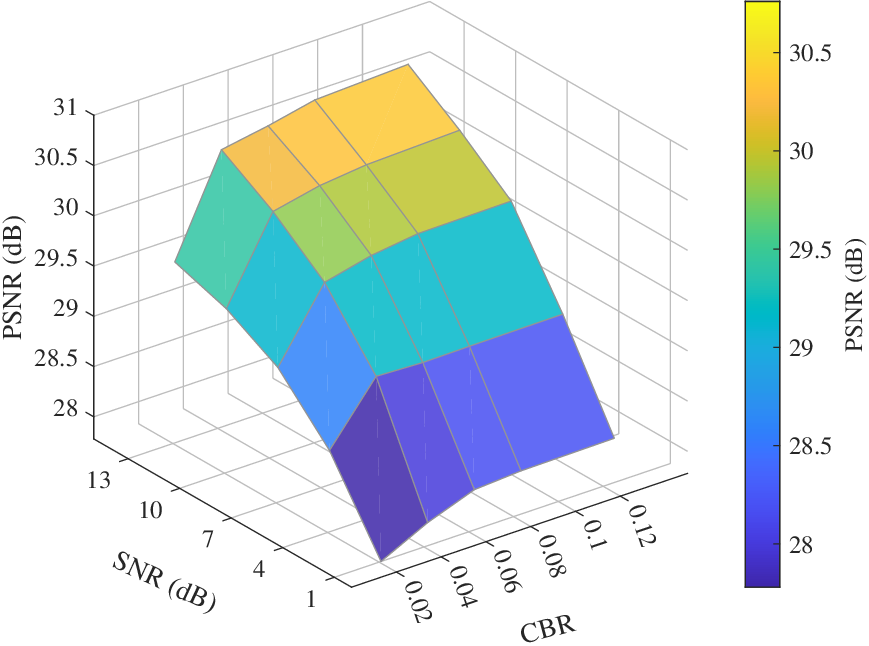}
        \caption{}
        \label{fig:sub4}
    \end{subfigure}
    
    \caption{Performance evaluation under different channel and rate conditions. 
(a) PSNR versus CBR at $\text{SNR}=13~\text{dB}$. 
(b) PSNR versus SNR at $\text{CBR}=0.125$. 
(c) MS-SSIM versus SNR at $\text{CBR}=0.125$. 
(d) PSNR surface over joint SNR and CBR variations.}
    \label{fig:SFMA performance}
\end{figure}

Fig.~\ref{fig:SFMA performance} and Fig.~\ref{fig:reconstruction} evaluate the reconstruction performance of the proposed SC-SFMA model against all baseline schemes. Fig.~\ref{fig:SFMA performance}(a) shows the PSNR versus CBR at $\text{SNR}=13$~dB, where SC-SFMA achieves competitive PSNR especially at low CBR values, confirming the effectiveness of rate-adaptive channel coding in bandwidth-constrained regimes. Fig.~\ref{fig:SFMA performance}(b) plots the PSNR versus SNR at $\text{CBR}=0.125$. SC-SFMA achieves approximately $28$--$31$~dB across the tested range of $[0,13]$~dB, maintaining stable performance at low SNR through similarity-conditioned feature fusion. In contrast, BPG+LDPC+QAM exhibits the cliff effect below $1$~dB SNR, while its NOMA variant fails entirely due to compounded channel noise and inter-user interference. SwinJSCC-OMA achieves the highest single-user PSNR ($31$--$37$~dB) by using dedicated orthogonal subchannels, but at the cost of reduced spectral efficiency. The DeepJSCC-NOMA and SwinJSCC+NOMA baselines, which employ physical-layer superposition without semantic feature fusion, achieve lower PSNR than SC-SFMA, demonstrating the advantage of the SCFM module for interference management.
Fig.~\ref{fig:SFMA performance}(c) confirms these trends in the perceptual domain: the MS-SSIM versus SNR curve at $\text{CBR}=0.125$ is consistent with the PSNR results, indicating that SC-SFMA preserves both pixel-level fidelity and structural quality. Fig.~\ref{fig:SFMA performance}(d) presents the PSNR surface over the joint SNR-CBR space. The surface ranges from approximately $28$~dB at low SNR and low CBR to over $31$~dB at high SNR and high CBR, with diminishing returns in CBR at high values. 
Fig.~\ref{fig:reconstruction} illustrates the qualitative reconstruction results of four baselines under identical channel conditions. The figure shows that SC-SFMA achieves the highest MS-SSIM and PSNR scores among all baselines and consistently delivers superior perceptual fidelity.

\begin{figure}
    \centering
    \includegraphics[width=\linewidth]{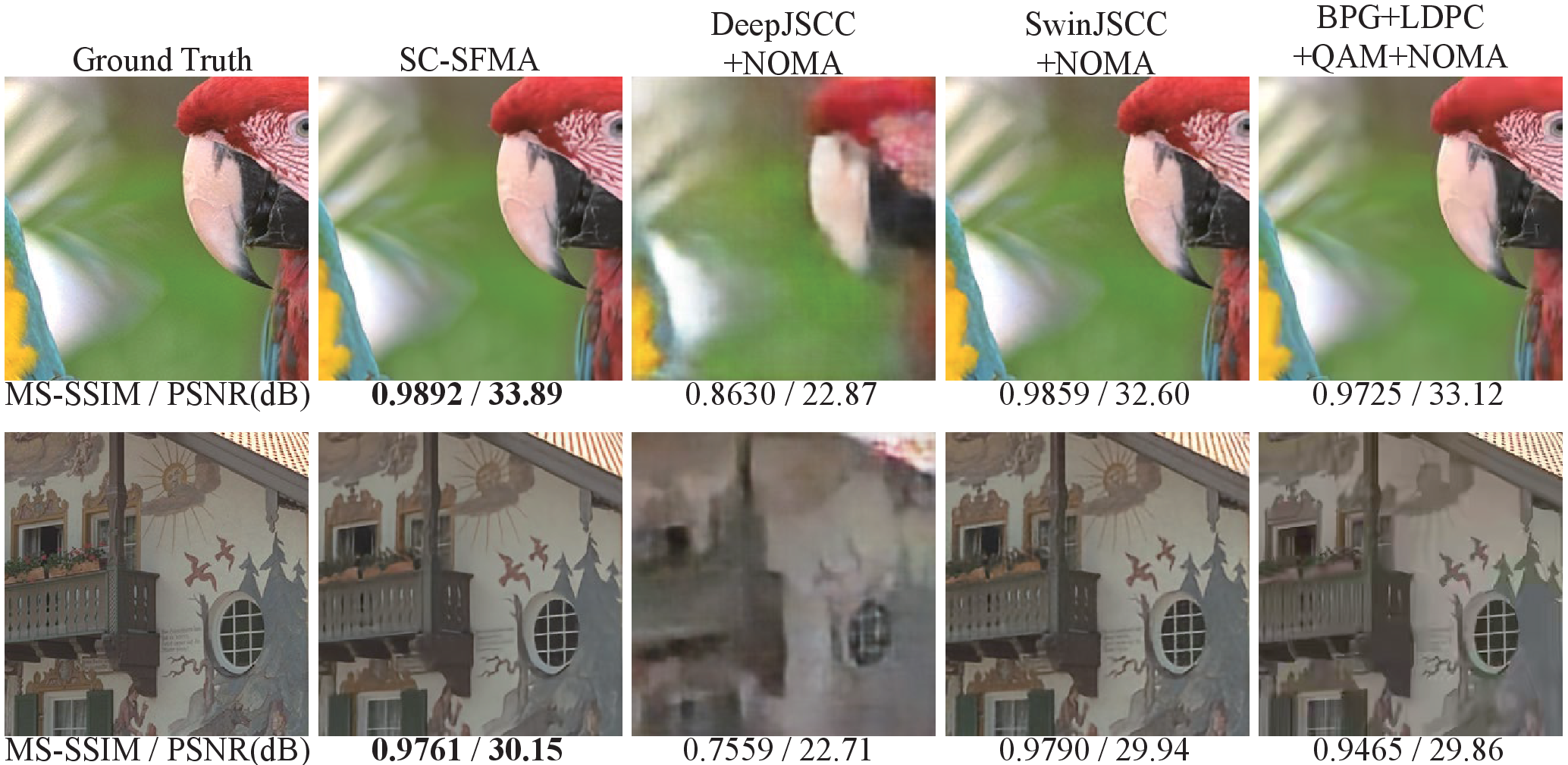}
    \caption{The reconstruction quality when $\text{SNR}=13$~dB and $\text{CBR}=0.125$ under different baselines. Each row corresponds to the reconstructed image of a single user across all baselines, while each column displays the concurrently transmitted images of two users.}
    \label{fig:reconstruction}
\end{figure}

\subsection{Joint Optimisation Performance}
\label{subsec:opt_performance}

\subsubsection{Sum Transmission Rate versus Transmit Power}

Fig.~\ref{fig:fig7_rate_vs_Pmax} plots the sum transmission rate versus $P^{\max}$ with $N = 10$ and $B^{\max} = 10$~MHz. It can be observed that the proposed SC-SFMA consistently outperforms all baseline schemes across the entire power range, achieving $90.3$~Mbps at $P^{\max} = 30$~dBm. This represents gains of $32\%$ over channel-based pairing ($68.3$~Mbps), $39\%$ over SwinJSCC+NOMA ($65.1$~Mbps), $19\%$ over SwinJSCC-FDMA ($76.0$~Mbps), and $12\%$ over equal allocation ($80.7$~Mbps).
The $39\%$ gain over SwinJSCC+NOMA directly quantifies the benefit of the SCFM module, since the two schemes share the same optimisation algorithms and differ only in the profiled $\rho_{ij}$. The proposed scheme also outperforms FDMA, confirming that the twice spectral reuse of SC-SFMA compensates for the residual interference when sufficient power is available. At low power ($P^{\max} = 20$~dBm), however, this margin narrows to $3\%$, because $\rho_{ij}(p_k, \delta_k)$ cannot approach its saturation floor $\rho^{\min}_{ij}$ at reduced per-group power, consistent with the double-saturation behaviour of the logistic model~\eqref{eq:rho_power}.
The performance gap between the proposed scheme and equal allocation ranges from $52\%$ at $20$~dBm to $2\%$ at $34$~dBm. This is because at high power, $\rho_{ij}$ saturates near $\rho^{\min}_{ij}$ for all groups, diminishing the marginal benefit of non-uniform allocation, and joint optimisation is therefore most valuable in the power-limited regime. Channel-based pairing remains below FDMA across the entire range despite using the SC-SFMA transceiver, confirming that the NOMA pairing criterion is counterproductive in SFMA: pairing by channel disparity rather than semantic similarity produces high $\rho_{ij}$ that negates the SCFM gain.

\begin{figure}
    \centering
    \includegraphics[width=0.8\linewidth]{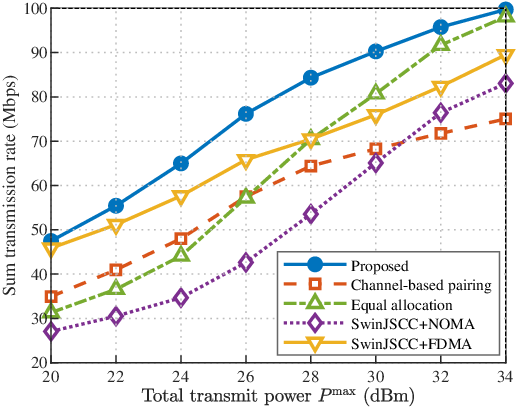}
    \caption{Average total transmission rate versus total transmit power $P^{\max}$ with $N=10$ users and $B^{\max}=10$~MHz.}
    \label{fig:fig7_rate_vs_Pmax}
\end{figure}

\begin{figure}
    \centering
    \includegraphics[width=0.8\linewidth]{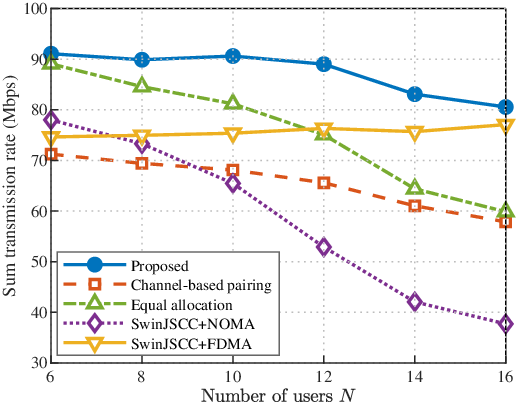}
    \caption{Average total transmission rate versus number of users $N$ with $P^{\max}=30$~dBm and $B^{\max}=10$~MHz.}
    \label{fig:fig8_rate_vs_N}
\end{figure}

\subsubsection{Total Transmission Rate versus Number of Users}

Fig.~\ref{fig:fig8_rate_vs_N} shows the sum transmission rate versus $N$ with $P^{\max} = 30$~dBm and $B^{\max} = 10$~MHz. As $N$ increases, the per-group power and bandwidth decrease as $p_k = 2P^{\max}\!/N$ and $b_k = 2B^{\max}\!/N$. The proposed algorithm exhibits the slowest degradation, declining from $91.1$~Mbps at $N = 6$ to $80.6$~Mbps at $N = 16$ ($-12\%$), whereas SwinJSCC+NOMA drops from $78.0$ to $37.7$~Mbps ($-52\%$). This resilience stems from the growing combinatorial space $\binom{N}{2}$ of candidate pairs, which provides the optimiser with more high-similarity options to compensate for the per-group resource dilution.
A notable feature is the approximately constant FDMA sum rate ($75$--$77$~Mbps): under equal per-user allocation, the SNR $P^{\max}|h_u|^2\!/(B^{\max} N_0)$ 
is independent of $N$, so the sum rate depends only on the average channel quality. The proposed scheme outperforms FDMA by $22\%$ at $N = 6$ and $+5\%$ at $N = 16$, with the crossover occurring around $N \approx 18$. This reflects a fundamental trade-off: SC-SFMA gains twice spectral reuse at the cost of semantic interference, and when per-group resources are too thin to suppress $\rho_{ij}$, the orthogonal strategy prevails. SC-SFMA is therefore most beneficial for moderate user densities.

\subsubsection{Total Transmission Rate versus Bandwidth}

\begin{figure}
    \centering
    \includegraphics[width=0.8\linewidth]{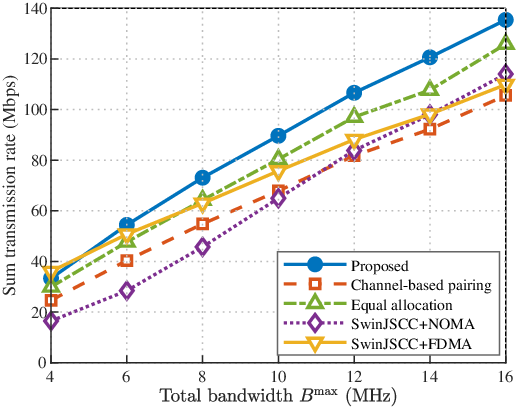}
    \caption{Average total transmission rate versus total bandwidth $B^{\max}$ with $N=10$ users and $P^{\max}=30$~dBm.}
    \label{fig:fig9_rate_vs_Bmax}
\end{figure}

Fig.~\ref{fig:fig9_rate_vs_Bmax} depicts the sum transmission rate versus $B^{\max}$ with $N = 10$ and $P^{\max} = 30$~dBm. All schemes exhibit near-linear growth. The proposed algorithm achieves $33.3$~Mbps at $B^{\max} = 4$~MHz and $135.4$~Mbps at $B^{\max} = 16$~MHz, outperforming FDMA by $23\%$, channel-based pairing by $28\%$, and SwinJSCC+NOMA by $19\%$ at $B^{\max} = 16$~MHz.
At low bandwidth ($B^{\max} = 4$~MHz), FDMA marginally exceeds the proposed scheme ($36.0$ versus $33.3$~Mbps). This is because each group receives only $b_k = 0.8$~MHz, pushing the system into a high-SNR regime where even moderate $\rho_{ij}$ causes significant interference relative to the small noise floor $b_k N_0$. The proposed scheme overtakes FDMA at $B^{\max} \approx 6$~MHz, and the gap widens monotonically thereafter, reaching $+23\%$ at $16$~MHz.
This widening can be attributed to the fact that additional bandwidth increases $b_k N_0$, shifting the operating point toward a noise-limited regime that amplifies the spectral reuse benefit. SwinJSCC+NOMA also surpasses FDMA at $B^{\max} = 16$~MHz ($114.0$ versus $109.9$~Mbps), indicating that the twice bandwidth advantage can compensate for higher interference when sufficient bandwidth dilutes its impact.

\subsubsection{Pairing Strategy Comparison}

\begin{figure}
    \centering
    \includegraphics[width=0.8\linewidth]{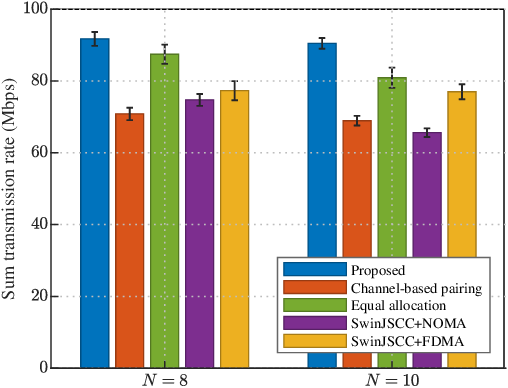}
    \caption{Pairing strategy comparison under different numbers of users with $P^{\max}=30$~dBm and $B^{\max}=10$~MHz.}
    \label{fig:fig10_ablation}
\end{figure}

Fig.~\ref{fig:fig10_ablation} compares all schemes for $N = 8$ and $N = 10$ at $P^{\max} = 30$~dBm and $B^{\max} = 10$~MHz, with error bars indicating one standard deviation. For $N = 10$, the proposed scheme achieves $90.5$~Mbps, outperforming equal allocation ($80.9$~Mbps, $+12\%$), SwinJSCC-FDMA ($77.0$~Mbps, $+18\%$), channel-based pairing ($68.9$~Mbps, $+31\%$), and SwinJSCC+NOMA ($65.6$~Mbps, $+38\%$).
By comparing the proposed scheme with equal allocation, we observe the gain from joint power-bandwidth optimisation: $5\%$ at $N = 8$ and $12\%$ at $N = 10$, consistent with the finding in Fig.~\ref{fig:fig7_rate_vs_Pmax} that resource optimisation is more valuable when per-group resources are scarcer. Channel-based pairing ranks last at $N = 10$ ($68.9$~Mbps), falling below both SwinJSCC+NOMA and FDMA despite using the SC-SFMA transceiver, highlighting that an inappropriate pairing strategy can negate the SCFM gain entirely. The standard deviation of the proposed scheme ($7.6$~Mbps) is approximately half that of equal allocation ($14.4$~Mbps), confirming that joint optimisation also improves robustness to channel and similarity variations.

\section{Conclusion}
\label{sec:conclusion}

In this paper, we have proposed the SC-SFMA framework for joint user pairing and resource allocation in multi-user semantic communication networks. A Swin Transformer-based transceiver with a dual-conditioned similarity modulator has been developed to adaptively fuse cross-user semantic features. The pair-dependent semantic interference has been modelled as a bivariate logistic function of transmit power and compression ratio, bridging the learned transceiver with network-level optimisation. To solve the resulting mixed-integer non-convex problem, a three-block alternating optimisation algorithm has been developed, comprising dual-decomposition-assisted compression ratio allocation, trust-region SCA for power–bandwidth optimisation, and a user pairing optimisation. Numerical results have shown the effectiveness of the proposed algorithm.
Extending the framework to multi-cell scenarios and incorporating generative prior-aided receivers are interesting directions for future work.

\bibliographystyle{IEEEtran}
\bibliography{ref}

\end{document}